\documentclass[sigconf,nonacm]{acmart}

\AtBeginDocument{%
  \providecommand\BibTeX{{%
    \normalfont B\kern-0.5em{\scshape i\kern-0.25em b}\kern-0.8em\TeX}}}

\usepackage{amsmath}
\usepackage{amssymb}
\usepackage{amsthm}
\usepackage{color}
\usepackage{bm}
\usepackage{hyperref}
\usepackage{here}
\usepackage{url}
\usepackage{subfig}
\usepackage{algpseudocode}
\usepackage[linesnumbered,ruled,vlined]{algorithm2e}
\MakeRobust{\Call}
\SetKwProg{Procedure}{Procedure}{}{}
\SetKwInOut{Input}{Input}
\SetKwInOut{Output}{Output}
\allowdisplaybreaks
\usepackage{booktabs}
\usepackage{subfig}

\newtheorem{thm}{Theorem}

\newtheorem{prop}[thm]{Proposition}
\newtheorem{lem}[thm]{Lemma}

\theoremstyle{remark}

\newcommand{\ds}{\displaystyle}

\newcounter{num}
\newcommand{\ZZ}{\mathbb{Z}} 
\newcommand{\RR}{\mathbb{R}} 
\newcommand{\ca}[1]{\mathcal{#1}} 

\newcommand{\ol}[1]{\overline{#1}} 
\newcommand{\wt}[1]{\widetilde{#1}} 
\newcommand{\vol}{\mathrm{vol}}
\newcommand{\sweep}{\mathrm{sweep}}
\newcommand{\1}{\mathbf{1}}

\newcommand{\argmin}{\mathop{\mathrm{argmin}}}
\newcommand{\argmax}{\mathop{\mathrm{argmax}}}

\newcommand{\pr}{\bm{\mathrm{pr}}}
\def\inprod<#1>{\left\langle #1 \right\rangle} 

\begin{document}

\title{Hypergraph Clustering Based on PageRank} 

\author{Yuuki Takai}
\authornote{Both authors contributed equally to this research.}
\affiliation{
\institution{RIKEN AIP}
\city{Tokyo}
\country{Japan}
}
\email{yuuki.takai@riken.jp}

\author{Atsushi Miyauchi}
\authornotemark[1]
\affiliation{
\institution{University of Tokyo}
\city{Tokyo}
\country{Japan}
}
\email{miyauchi@mist.i.u-tokyo.ac.jp}

\author{Masahiro Ikeda}
\affiliation{
\institution{RIKEN AIP}
\city{Tokyo}
\country{Japan}
}
\email{masahiro.ikeda@riken.jp}

\author{Yuichi Yoshida}
\affiliation{%
  \institution{National Institute of Informatics}
  \city{Tokyo}
  \country{Japan}
}
\email{yyoshida@nii.ac.jp}

\begin{abstract}
  A hypergraph is a useful combinatorial object to model ternary or higher-order relations among entities.
  Clustering hypergraphs is a fundamental task in network analysis.
  In this study, we develop two clustering algorithms based on personalized PageRank on hypergraphs.
  The first one is local in the sense that its goal is to find a tightly connected vertex set with a bounded volume including a specified vertex.
  The second one is global in the sense that its goal is to find a tightly connected vertex set.
  For both algorithms, we discuss theoretical guarantees on the conductance of the output vertex set.
  Also, we experimentally demonstrate that our clustering algorithms outperform existing methods in terms of both the solution quality and running time.
  To the best of our knowledge, ours are the first practical algorithms for hypergraphs with theoretical guarantees on the conductance of the output set.
\end{abstract}



\maketitle


\section{Introduction}

Graph clustering is a fundamental task in network analysis, where the goal is to find a tightly connected vertex set in a graph.
It has been used in many applications such as community detection~\cite{Fortunato:2010iw} and image segmentation~\cite{Felzenszwalb:2004bx}.
Because of its importance, graph clustering has been intensively studied, and a plethora of clustering methods have been proposed.
See, e.g.,~\cite{Fortunato:2010iw} for a survey.

There is no single best criterion to measure the quality of a vertex set as a cluster; however, one of the most widely used is the notion of \emph{conductance}.
Roughly speaking, the conductance of a vertex set is small if many edges reside within the vertex set, whereas only a few edges leave it (see Section~\ref{sec:pre} for details).
Many clustering methods have been shown to have theoretical guarantees on the conductance of the output vertex set~\cite{Alon:1986gz,Alon:1985jg,Andersen:2007tsa,Kloster:2014wq,Chung:2007ep,Spielman:2013cc,Gharan:2012fy,Andersen:2009fj,Kwok:2017ga}.

Although a graph is suitable to represent binary relations over entities, many relations, such as coauthorship, goods that are 
purchased together, and chemicals involved in metabolic reactions inherently comprise three or more entities. 
A hypergraph is a more natural 
object to model such relations. 
Indeed, if we represent coauthorship network by a graph, 
an edge only indicates the existence of co-authored paper, and information about co-authors in one paper is lost. 
However, hypergraphs do not lose that information.



As with graphs, hypergraph clustering is also an important task,  
and it is natural to use an extension of conductance for hypergraphs~\cite{Chan:2018eu,Yoshida:2019zz} to measure the quality of a vertex set as a cluster.
Although several methods have been shown to have theoretical guarantees on the conductance of the output vertex set~\cite{Chan:2018eu,Yoshida:2019zz,ikeda2018finding}, none of them are efficient in practice.

To obtain efficient hypergraph clustering algorithms with a theoretical guarantee, we utilize the graph clustering algorithm developed by Andersen~\emph{et~al.}~\cite{Andersen:2007tsa}.
Their algorithm is based on \emph{personalized PageRank} (PPR), which is the stationary distribution of a random walk that jumps back to a specified vertex with a certain probability in each step of the walk.
Roughly speaking, their algorithm first computes the PPR from a specified vertex, and then returns a set of vertices with high PPR values.
They showed a theoretical guarantee on the conductance of the returned vertex set.

In this work, we propose two hypergraph clustering algorithms by extending Andersen~\emph{et~al.}'s algorithm~\cite{Andersen:2007tsa}.
The first one is \emph{local} in the sense that its goal is to find a vertex set of a small conductance and a bounded volume including a specified vertex.
The second one is \emph{global} in the sense that its goal is to find a vertex set of the smallest conductance.

Although there could be various definitions of random walk on hypergraphs and corresponding PPRs, we adopt the PPR recently introduced by Li and Milenkovic~\cite{Li:2018we} using a hypergraph Laplacian~\cite{Chan:2018eu,Yoshida:2019zz}.
Because the hypergraph Laplacian well captures the cut information of hypergraphs, the corresponding PPR can be used to find vertex sets with small conductance.

Because we can efficiently compute (a good approximation to) the PPR, our algorithms are efficient.
Moreover, as with Andersen~\emph{et~al.}'s algorithm~\cite{Andersen:2007tsa}, our algorithms have theoretical guarantees on the conductance of the output vertex set.
To the best of our knowledge, no local clustering algorithm with a theoretical guarantee on the conductance of the output set is known for hypergraphs.
Although, as mentioned previously, several algorithms have been proposed for global clustering for hypergraphs~\cite{Chan:2018eu,Yoshida:2019zz,ikeda2018finding}, we stress here that our global clustering algorithm is the first practical algorithm with a theoretical guarantee.

We experimentally demonstrate that our clustering algorithms outperform other baseline methods in terms of both the solution quality and running time.

Our contribution can be summarized as follows.
\begin{itemize}
\leftskip=-15pt
\itemsep=0pt
\item We propose local and global hypergraph clustering algorithms based on PPR for hypergraphs given by Li and Milenkovic~\cite{Li:2018we}.
\item We give theoretical guarantees on the conductance of the vertex set output by our algorithms.
\item Our algorithms are more efficient than existing algorithms with theoretical guarantees~\cite{Chan:2018eu,Yoshida:2019zz,ikeda2018finding} 
(see Section~\ref{sec:related} for the details).
\item We experimentally demonstrate that our clustering algorithms outperform other baseline methods in terms of both the solution quality and running time.
\end{itemize}

This paper is organized as follows.
We discuss related work in Section~\ref{sec:related} and then introduce notions used throughout the paper in Section~\ref{sec:pre}.
After discussing basic properties of PPR for hypergraphs in Section~\ref{sec:ppr}, we introduce a key lemma
(Lemma~\ref{lem:upperbound-of-sweep-conductance-local-2}) that connects conductance and PPR in Section~\ref{sec:conductance}.
On the basis of the key lemma, we show our local and global clustering algorithms in Sections~\ref{sec:local-clustering} and~\ref{sec:global-clustering}, respectively.
We discuss our experimental results in Section~\ref{sec:experiments}, and conclude in Section~\ref{sec:conclusion}. In appendix, we 
give the proofs which could not be included in the main body 
because of the page restrictions.


\section{Related Work}\label{sec:related}

As previously mentioned, even if we restrict our attention to graph clustering methods with theoretical guarantees on the conductance of the output vertex set, a plethora of methods  have been proposed~\cite{Alon:1986gz,Alon:1985jg,Andersen:2007tsa,Kloster:2014wq,Chung:2007ep,Spielman:2013cc,Gharan:2012fy,Andersen:2009fj,Kwok:2017ga}.
Some of them run in time nearly linear in the size of the output vertex set, i.e., we do not have to read the entire graph.

Although several hypergraph clustering algorithms have been proposed~\cite{Liu:2010va,Zhou:2006vj,Leordeanu:2012tj,RotaBulo:2013da,Agarwal:2005go}, most of them are merely heuristics and have no guarantees on the conductance of the output vertex set.
Indeed, we are not aware of any local clustering algorithms with a theoretical guarantee, and to the best of our knowledge, there are only two global clustering algorithms with theoretical guarantees.

The first one~\cite{Chan:2018eu,Yoshida:2019zz} is based on Cheeger's inequality for hypergraphs.
Given a hypergraph $H=(V,E,w)$, it first computes an approximation $\bm{x} \in \mathbb{R}^V$ to the second eigenvector of the normalized Laplacian $\mathcal{L}_H$ (see Section~\ref{subsec:pagerank-hypergraph} for details), and then returns a sweep cut, i.e., a vertex set of the form $\{v \in V : \bm{x}(v) > \tau\}$ for some $\tau \in \mathbb{R}$.
We can guarantee that the output vertex set has conductance roughly $O\left(\sqrt{\phi_H \log n}\right)$, where $\phi_H$ is the conductance of the hypergraph (see Section~\ref{subsec:hypergraph}).
However, to run the algorithm we need to approximately compute the second eigenvalue of the normalized Laplacian, which requires solving SDP\@; hence, is impractical.

The second algorithm was developed by 
Ikeda~\emph{et~al.}~\cite{ikeda2018finding}.
Their algorithm is based on simulating the heat equation:
\[
  \frac{d\bm{\rho}_t}{d t} \in  -\mathcal{L}_H(\bm{\rho}_t), \quad \bm{\rho}_0 = \bm{\chi}_v, 
\]
where $\bm{\chi}_v \in \mathbb{R}^V$ is a vector with $\bm{\chi}_v(u) = 1$ if $u = v$ and $0$ otherwise.
Then, they showed that one of the sweep cuts of the form $\{v \in V : \bm{\rho}_t(v) > \tau\}$ for some $\tau \in \mathbb{R}$ and $t \geq 0$, has conductance roughly $O\left(\sqrt{\phi_H}\right)$, under a mild condition.
A drawback of their algorithm is that we need to calculate sweep cuts induced by $\bm{\rho}_t$ for every $t \geq 0$, which cannot be efficiently implemented on Turing machines.
Although our algorithm also computes sweep cuts, we need only consider one specific vector, i.e., the PPR\@.

If we do not need any theoretical guarantee on the quality of the output vertex set, we can think of a heuristic in which we first construct a graph from the input hypergraph $H=(V,E,w)$ and then apply a known clustering algorithm for graphs on the resulting graph, e.g., Andersen~\emph{et~al.}'s algorithm~\cite{Andersen:2007tsa}.
There are two popular methods for constructing graphs:
\begin{description}
\leftskip=-5pt
\item[Clique expansion~\cite{Agarwal:2005go,Zhou:2006vj}:] For each hyperedge $e \in E$ and each $u,v \in e$ with $u \neq v$, we add an edge $uv$ of weight $w(e)$ or $w(e)/\binom{|e|}{2}$.
\item[Star expansion~\cite{Zien:1999cm}:] For each hyperedge $e \in E$, we introduce a new vertex $v_e$, and then for each vertex $u \in e$, we add an edge $uv_e$ of weight $w(e)$ or $w(e)/|e|$.
\end{description}
The obvious drawback of the clique expansion is that it introduces $\Theta(k^2)$ edges for a hyperedge of size $k$, and hence the resulting graph is huge.
The relation of these expansions and various Laplacians defined for hypergraphs~\cite{bolla1993,Rodriguez:2002bz,Zhou:2005tc} were discussed in~\cite{Agarwal:2006ci}.

Ghoshdastidar and Dukkipati~\cite{Ghoshdastidar:2014wv} proposed a spectral method utilizing a tensor defined from the input hypergraph, and 
analyzed its performance 
when the hypergraph is generated from
 a planted partition model or a stochastic block model.
A drawback of their method is that the 
hypergraph must be \emph{uniform}, i.e., each hyperedge must be of the same size.
Chien~\emph{et~al.}~\cite{Chien:2018up} proposed another statistical method that also assumes uniformity of the 
hypergraph, and analyzed its performance on a stochastic block model.

Other iterative procedures for clustering hypergraphs have also been proposed~\cite{Leordeanu:2012tj,RotaBulo:2013da,Liu:2010va}.
However, these methods do not have any theoretical guarantee on the quality of the output vertex set.

\section{Preliminaries}\label{sec:pre}

We say that a vector $\bm{x} \in \mathbb{R}^V$ is a \emph{distribution} if $\sum_{v \in V}\bm{x}(v) = 1$ and $\bm{x}(v) \geq 0$ for every $v \in V$.
For a vector $\bm{x} \in \RR^V$ and a set $S \subseteq V$, we define $\bm{x}(S) = \sum_{v \in S}\bm{x}(v)$.

We often use the symbols $n$ and $m$ to denote the number of vertices and number of (hyper)edges of the (hyper)graph we are concerned with, which should be clear from the context.

\subsection{Personalized PageRank}\label{subsec:pagerank-graph}
Let $G=(V,E,w)$ be a weighted graph, where $w\colon E \to \mathbb{R}_+$ is a weight function over edges.
For convenience, we define $w(uv) = 0$ when $uv \not \in E$.
The \emph{degree} of a vertex $u \in V$, denoted by $d_u$, is
$\sum_{e\in E:u\in e}w(e)$. We assume $d_u>0$ for any $u\in V$.
The degree matrix $D_G \in \mathbb{R}^{V \times V}$ and adjacency matrix $A_G \in \mathbb{R}^{V \times V}$ are defined as
$D_G(u,u)=d_u$ and $D_G(u,v) = 0$ for $u\neq v$, and
$A_G(u,v) = w(uv)$.
\emph{Lazy random walk} is a stochastic process, which stays at the current vertex $u \in V$ with probability $1/2$ and moves to a neighbor $v \in V$ with probability $1/2 \cdot w(uv)/d_u$.
The \emph{transition matrix} of a lazy random walk can be written as $W_G = (I + A_G D_G^{-1})/2$.

For a graph $G=(V,E,w)$, a seed distribution $\bm{s} \in \mathbb{R}^V$ and a parameter $\alpha \in (0,1]$, the \emph{personalized PageRank} (PPR) is a vector $\mathbf{pr}_\alpha(\bm{s}) \in \mathbb{R}^V$~\cite{Jeh:2003kz} defined as the stationary distribution of a random walk that, with probability $\alpha$, jumps to a vertex $v \in V$ with probability $\bm{s}(v)$, and with the remaining probability $1-\alpha$, moves as the lazy random walk.
We can rephrase $\mathbf{pr}_\alpha(\bm{s}) \in \mathbb{R}^V$ as the solution to the following equation:
\[
  \mathbf{pr}_\alpha(\bm{s}) = \alpha \bm{s} + (1-\alpha)W_G \mathbf{pr}_\alpha(\bm{s}).
\]
Using this equation, we can also define PPR for any vector
$\bm{s} \in \mathbb{R}^V$.

For later convenience, we rephrase PPR in terms of \emph{Laplacian}.
For a graph $G=(V,E,w)$, its Laplacian $L_G\in \mathbb{R}^{V \times V}$ and (random-walk) normalized Laplacian $\mathcal{L}_G\in \mathbb{R}^{V \times V}$ are defined as
  $L_G = D_G - A_G$
  \text{and}
  $\mathcal{L}_G = L_G D_G^{-1} = I - A_G D_G^{-1}$,
respectively. Then, we have $W_G = I - \ca{L}_G/2$.
It is easy to check that the PPR $\bm{\pr}_\alpha(\bm{s})$ satisfies the equation 
\begin{align}
   \left(I + \frac{1-\alpha}{2\alpha}\ca{L}_G\right) \bm{\pr}_\alpha(\bm{s}) =  \bm{s}. \label{eq:pagerank-via-laplacian}
\end{align}

\subsection{Hypergraphs}\label{subsec:hypergraph}
In this section, we introduce several notions for hypergraphs.
We note that these notions can also be used for graphs, because a hypergraph is a generalization of a graph.

Let $H = (V, E, w)$ be a (weighted) hypergraph, where $w\colon E \to \mathbb{R}_+$ is a weight function over hyperedges.
The degree of a vertex $u \in V$, denoted by $d_u$, is $\sum_{e \in E: u \in e}w(e)$. We assume $d_u>0$ for any $u\in V$.
The degree matrix $D_H \in \mathbb{R}^{V \times V}$ is defined as
$D_H(u,u) = d_u$ and $D_H(u,v) = 0$ for $u\neq v$.
Let $S \subseteq V$ be a vertex set and
$\ol{S} = V\!\setminus\! S$. We define the set of the
boundary hyperedges $\partial S$ by $\partial S = \left\{ e\in E \mid e\cap S \neq \emptyset \  \mbox{and} \ e\cap \ol{S} \neq \emptyset \right\}$.
We define $S^\circ$ as the \emph{interior} of $S$, that is, $v \in S^\circ$ if and only if $v \in S$ and $v \not\in e$ for any 
$e \in\partial S$.

We define $\mathrm{cut}(S)$ as the total weight of hyperedges intersecting both $S$ and $\ol{S}$, that is, $\mathrm{cut}(S) = \sum_{e \in \partial S}w(e)$.
We say that $H$ is \emph{connected} if $\mathrm{cut}(S) > 0$ for every $\emptyset \subsetneq S \subsetneq V$.

Next, we define the \emph{volume} of a vertex set $S$ as $\mathrm{vol}(S) = \sum_{v \in V}d_v$.
Then, we define the \emph{conductance} of $S$ as
\[
  \phi_H(S) = \frac{\mathrm{cut}(S)}{\min\{\mathrm{vol}(S),\mathrm{vol}(V \setminus S)\}}.
\]
We can say that a smaller $\phi_H(S)$ implies that the vertex set $S$ is a better cluster.
The \emph{conductance} of $H$ is $\phi_H = \min_{\emptyset \subsetneq S \subsetneq V}\phi_H(S)$.

Finally for $S\subseteq V$, we define the distribution
$\bm{\pi}_S \in \mathbb{R}^V$ as
\[
 \bm{\pi}_S(v) = \begin{cases}
  \frac{d_v}{\vol(S)} & \text{if} \ v \in S, \\
  0 & \text{otherwise.}
 \end{cases}
\]
We define $\bm{\chi}_v$ as the indicator vector for $v\in V$, i.e.,
$\bm{\chi}_v =\bm{\pi}_{\{v \}}$.

\subsection{Personalized PageRank for hypergraphs}\label{subsec:pagerank-hypergraph}
In this section, we explain the notion of PPR for hypergraphs recently introduced by Li and Milenkovic~\cite{Li:2018we}.

We start by defining the Laplacian for hypergraphs and then use~\eqref{eq:pagerank-via-laplacian} to define PPR\@.
For a hypergraph $H = (V, E, w)$, the \emph{(hypergraph) Laplacian} $L_H\colon \RR^V \to 2^{\RR^V}$ of $H$ is defined as
\begin{align}
 L_H(\bm{x}) = \left. \left\{ \ds\sum_{e \in E} w(e)
 \bm{b}_e\bm{b}_e^\top \bm{x} \ \right| \ \bm{b}_e \in \argmax_{\bm{b}\in B_e} \bm{b}^\top \bm{x} \right\} \subseteq \RR^V, \label{eq:hypergraph-laplacian}
\end{align}
where $B_e$ for $e\in E$ is the convex hull of the vectors $\{ \bm{\chi}_u - \bm{\chi}_v \mid u,v \in e \}$~\cite{Chan:2018eu,Yoshida:2016ig,Yoshida:2019zz}.
Note that $L_H$ is no longer a matrix, as opposed to the Laplacian of graphs.
We then define the \emph{normalized Laplacian} $\mathcal{L}_H\colon \mathbb{R}^V \to 2^{\mathbb{R}^V}$ as $\bm{x} \mapsto L_H(D_H^{-1}\bm{x})$.
We also define $A_H\colon \mathbb{R}^V \to 2^{\mathbb{R}^V}$ as $A_H = D_H-L_H$, or more formally, $\bm{x} \mapsto \{D_H\bm{x} - \bm{y} \mid \bm{y} \in L_H(\bm{x})\}$ and $W_H\colon \mathbb{R}^V \to 2^{\mathbb{R}^V}$ as
$\bm{x} \mapsto (\bm{x} + A_H(D_H^{-1}\bm{x}))/2$.
Note that when $H$ is a graph, each of $L_H(\bm{x})$, $\mathcal{L}_H(\bm{x})$, $A_H(\bm{x})$, and $W_H(\bm{x})$ defined above is a singleton consisting of the vector obtained by applying the corresponding matrix to $\bm{x}$.

To understand how the Laplacian $L_H$ acts on $\bm{x}$ when $H$ is not a graph, suppose that the entries of the vector $\bm{x} \in \mathbb{R}^V$ are pairwise distinct.
Then for each hyperedge $e \in E$, we have $\bm{b}_e = \bm{\chi}_{s_e} - \bm{\chi}_{t_e}$, where $s_e = \argmax_{v \in e}\bm{x}(v)$ and $t_e = \argmin_{v \in e}\bm{x}(v)$.
Then, we construct a graph $H_{\bm{x}} = (V,E_{\bm{x}},w_{\bm{x}})$ on the vertex set $V$ by, for each hyperedge $e \in E$, adding the edge $s_e t_e$ of weight $w(e)$ and the self-loop of weight $w(e)$ at $v $ for each $v \in e \setminus \{s_e,t_e\}$.
Note that the degree of a vertex is unchanged between $H$ and $H_{\bm{x}}$.
Then, $L_H(\bm{x}) = \{L_{H_{\bm{x}}}\bm{x}\}$, where $L_{H_{\bm{x}}}$ is the Laplacian of the graph $H_{\bm{x}}$.
An example of the construction of $H_{\bm{x}}$ is shown in Figure~\ref{fig:hypergraph}.

\begin{figure}[t!]
  \centering
  \includegraphics[width=0.85\hsize]{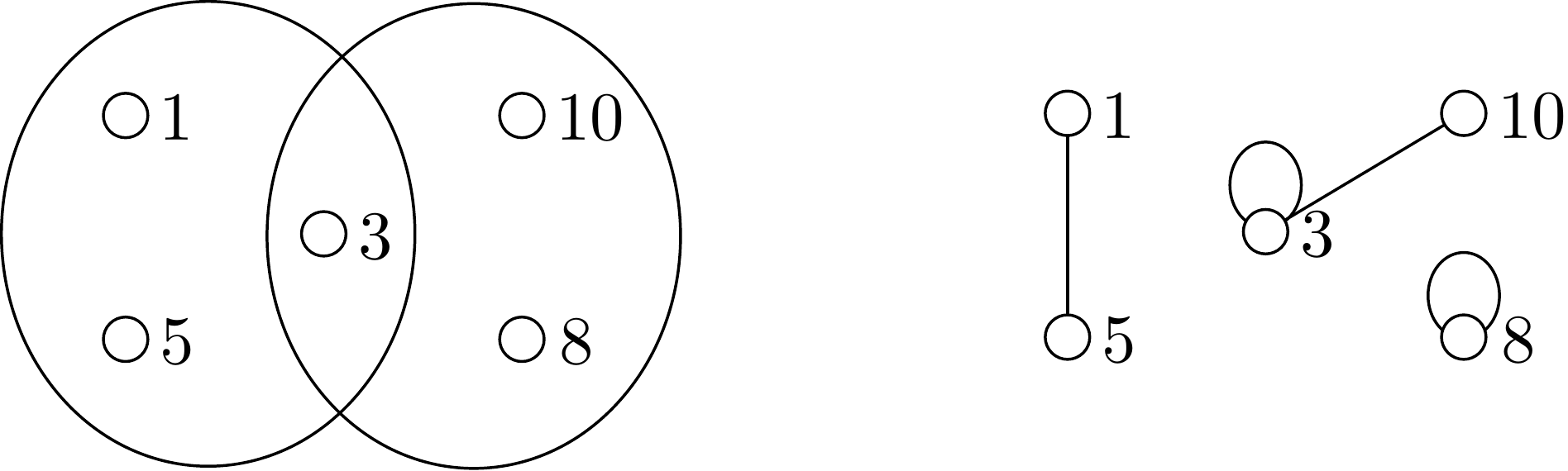}
  \caption{An example of a hypergraph $H$ and the graph $H_{\bm{x}}$. The values next to vertices represent a vector $\bm{x} \in \mathbb{R}^V$.}\label{fig:hypergraph}
\end{figure}

For general $\bm{x} \in \mathbb{R}^V$, there could be many choices for $\bm{b}_e\;(e\in E)$ in~\eqref{eq:hypergraph-laplacian}.
For any choice of $\bm{b}_e \in \argmax_{\bm{b} \in B_e}\bm{b}^\top \bm{x}$ for each $e \in E$, we can define a graph
$H_{\bm{x}} = (V,E_{\bm{x}},w_{\bm{x}})$ so that
$L_{H_{\bm{x}}} =  \sum_{e \in E}w(e) \bm{b}_e \bm{b}_e^\top$.
We remark that $w_{\bm{x}}$ can output negative values naturally.
However, it is easy to see that there exists a graph
$\wt{H}_{\bm{x}} = (V, \wt{E}_{\bm{x}}, \wt{w}_{\bm{x}})$ such that
$\wt{w}_{\bm{x}}$ is non-negative and
$\ca{L}_{\wt{H}_{\bm{x}}}\bm{x} = \ca{L}_{H_{\bm{x}}}\bm{x}$. This
$\wt{H}_{\bm{x}}$ is the same as the graph introduced in~\cite{Chan:2018eu}. In the following, we assume that when we choose
a graph $H_{\bm{x}}$, the weight function $w_{\bm{x}}$ is
non-negative.

Following~\eqref{eq:pagerank-via-laplacian}, we define the PPR $\bm{\pr}_\alpha(\bm{s})$ of $H$ as a solution (if exists) to
\begin{align}
   \left(I + \frac{1-\alpha}{2\alpha}\ca{L}_H\right)(\bm{x}) \ni  \bm{s}. \label{eq:hypergraph-pagerank-via-laplacian}
\end{align}
We can observe that $\bm{\pr}_\alpha(\bm{s}) = \bm{s}$ when $\alpha = 1$ and $\bm{\pr}_\alpha(\bm{s}) \to \bm{\pi}_V$ when $\alpha \to 0$.
In Section~\ref{sec:ppr}, we show that~\eqref{eq:hypergraph-pagerank-via-laplacian} indeed has a solution and it is unique for any $\bm{s} \in \mathbb{R}^V$ and $\alpha \in (0,1]$.

\subsection{Sweep cuts}
For a hypergraph $H=(V,E,w)$ and a vector $\bm{x} \in \mathbb{R}^V$, we say that a set of the form $\{v \in V \mid \bm{x}(v)/d_v > \tau\}$ for some $\tau \in \mathbb{R}$ is a \emph{sweep cut} induced by $\bm{x}$.

To discuss local clustering, we consider the minimum conductance of a sweep cut with a restriction on the volume.
First, we order the vertices $v_1, v_2, \dots, v_n$ in $V$ so that
\begin{align*}
\frac{\bm{x}(v_1)}{d_{v_1}} \geq \frac{\bm{x}(v_2)}{d_{v_2}} \geq \cdots \geq \frac{\bm{x}(v_n)}{d_{v_n}}.
\end{align*}
Then, we can represent sweep cuts explicitly as 
\begin{align*}
S_j^{\bm{x}} = \{ v_1, v_2, \dots, v_j \}\ \ (j = 1, 2, \dots, n).
\end{align*}
We set $S_0^{\bm{x}}=\emptyset$ for convenience.
For $\mu \in (0,1/2]$, let $\ell_\mu$ be the unique integer such that $\vol(S_{\ell_\mu-1}^{\bm{x}})< \mu \cdot \vol(V) \leq \vol(S_{\ell_\mu}^{\bm{x}})$.
Then, we define $\phi_H^\mu(\bm{x})$ as
$ \phi_H^\mu(\bm{x}) =
 \min\{\phi_H(S_j^{\bm{x}})\mid j = 1,2, \dots, \ell_\mu \}$.


\section{Personalized PageRank for Hypergraphs}\label{sec:ppr}
In this section, we discuss basic properties and the computational aspects of the PPR for hypergraphs. 
The proofs of the theorem and propositions in this section are 
given in Appendix~\ref{sec:proof-basic-properties}. 

\subsection{Basic properties}\label{sec:basic-properties}

Here we introduce some properties of the PPR. 
We first argue that~\eqref{eq:hypergraph-pagerank-via-laplacian} has a solution and it is unique for any $\bm{s} \in \mathbb{R}^V$ and $\alpha\in(0,1]$.
\begin{thm}\label{thm:existence-of-PR}
  For any $\bm{s}\in \RR^V$ and $\alpha \in (0,1]$,
  there exists a unique solution to~\eqref{eq:hypergraph-pagerank-via-laplacian}, and hence the PPR $\bm{\pr}_\alpha(\bm{s})$ is
  well-defined.
  Moreover, if $\bm{s}$ is a distribution, then
  so is $\pr_\alpha(\bm{s})$.
\end{thm}
%

We next investigate the continuity of $\bm{\pr}_\alpha(\bm{s})$ 
with respect to $\alpha$. 

\begin{prop}\label{prop:continuity}
  For any $\bm{s} \in \RR^{V}$, the PPR $\bm{\pr}_\alpha(\bm{s})$ is continuous with respect to $\alpha \in (0,1]$.
\end{prop}

The following proposition shows that the value of PPR at each vertex can be bounded by the value of the distribution $\bm{\pi}_V$.  
\begin{prop}\label{prop:bounds-of-values-of-PPR}
  Let $H=(V,E,w)$ be a connected hypergraph.
  Then for any $\alpha \in (0,1]$ and a vertex $u \in V$, we have
  \[
   \pr_\alpha(\bm{\chi}_u)(v) \begin{cases}
    \geq \bm{\pi}_V(u) & \mbox{if} \ v = u, \\
    \leq \bm{\pi}_V(v) & \mbox{if} \ v \neq  u.
   \end{cases}
  \]
\end{prop}
%
%


\subsection{Computation}
It is not difficult to show that we can compute the PPR in polynomial time by directly solving~\eqref{eq:hypergraph-pagerank-via-laplacian} using the techniques developed in~\cite{fujii2018polynomial}.
However, their algorithm involves the ellipsoid method, which is inefficient in practice.
Hence, we consider a different approach here.

Lemma~\ref{lem:pagerank-in-terms-of-beta} in Appendix~\ref{sec:proof-basic-properties} indicates that $\bm{\pr}_\alpha(\bm{s})$ is a stationary solution to the following differential equation:
\begin{align}
  \frac{d \bm{\rho}_t}{dt} \in \beta (\bm{s}- \bm{\rho}_t) - (1-\beta) \ca{L}_H (\bm{\rho}_t), 
  \label{eq:rho-t}
\end{align}
where $\beta:=2\alpha/(1+\alpha)$. 
From the fact that $\ca{L}_H$ is a maximal monotone operator~\cite{ikeda2018finding}, it is easy to show that the operator
\begin{align}
 \bm{x} \mapsto \beta (\bm{s}- \bm{x}) - (1-\beta) \ca{L}_H (\bm{x})
 \label{eq:pagerank-operator}
\end{align}
is also maximally monotone.
This implies that~\eqref{eq:rho-t} has a unique solution for any
given initial state $\bm{\rho}_0$.
Moreover, if the hypergraph $H$ is connected, then the solution $\bm{\rho}_t$ of~\eqref{eq:rho-t} converges to the stationary solution $\bm{\pr}_\alpha(\bm{s})$ as $t$ tends to infinity.
Hence, we can compute the PPR by simulating~\eqref{eq:rho-t}.

It is not difficult to show that we can simulate~\eqref{eq:rho-t} on a Turing machine with a certain error guarantee by time discretization using the techniques 
in~\cite{ikeda2018finding}.
In our experiments (Section~\ref{sec:experiments}), however, we simply simulated it using the Euler method for practical efficiency.

We remark that simulating~\eqref{eq:rho-t} by the Euler method 
is equivalent to running the gradient descent algorithm on the convex function
\begin{align*}
  Q(\wt{\bm{x}}) = \frac{\beta}{2}\|\wt{\bm{x}} - \wt{\bm{s}} \|^2 + (1-\beta) \wt{\bm{x}}^\top\ol{\ca{L}}_H (\wt{\bm{x}}),
\end{align*}
where $\wt{\bm{s}} = D^{-1/2}\bm{s}$ and
$\ol{\ca{L}}_H = D^{-1/2}L_H D^{-1/2}$.
To see this, note that the operator~\eqref{eq:pagerank-operator} can be rewritten as
the operator $\bm{x} \mapsto -D^{1/2}\partial Q(\wt{\bm{x}})$, where
$\wt{\bm{x}}= D^{-1/2}\bm{x}$ and $\partial Q(\wt{\bm{x}})$ is the
sub-differential of $Q(\wt{\bm{x}})$ with respect to the variable
$\wt{\bm{x}}$.
Then, the differential inclusion~\eqref{eq:rho-t} 
 can be written as
\begin{align}
 \frac{d \wt{\bm{\rho}}_t}{dt} \in - \partial Q(\wt{\bm{\rho}}_t)
 \label{eq:pde-gradient}
\end{align}
for $\wt{\bm{\rho}}_t = D^{-1/2}\bm{\rho}_t$ and by~\cite[Theorem~6.9]{peypouquet2010evolution},
the Euler sequence ${\{ \wt{\bm{\rho}}_{t_i} \}}_{i = 1}^\infty$ for~\eqref{eq:pde-gradient} converges to an optimal solution $\wt{\bm{\rho}}_\infty$, i.e., a solution satisfying
$\bm{0}\in \partial Q(\wt{\bm{\rho}}_\infty)$, under the assumption
$\sum\| \wt{\bm{\rho}}_{t_{i+1}} - \wt{\bm{\rho}}_{t_i} \|^2 <\infty$.
Then, $\bm{\rho}_\infty = D^{1/2}\wt{\bm{\rho}}_\infty$ is nothing
but the PPR $\pr_\alpha(\bm{s})$.


\section{Personalized PageRank and Conductance}\label{sec:conductance}
In this section, we introduce the following lemma, which extends the result in~\cite{Andersen:2007tsa} to hypergraphs and to the local setting. We prove this lemma in Appendix~\ref{sec:proof-of-key-lemma}. 
\begin{lem}\label{lem:upperbound-of-sweep-conductance-local-2}
Let $\bm{s}\in \mathbb{R}^V$ be a distribution and $\mu \in (0,1/2]$. 
  If there is a vertex set $S \subseteq V$ and
  a constant
  $\delta \geq 4/\sqrt{\vol(V)}$ such that
  $\vol(S)/\vol(V) \leq \mu$ and $\bm{\pr}_\alpha(\bm{s}) (S) -\bm{\pi}_V(S) > \delta$,
  then we have
  \begin{align}
    \phi_H^\mu(\bm{\pr}_{\alpha}(\bm{s})) <
    \sqrt{\frac{24 \alpha \log(4/\delta)}{\delta}}.\label{eq:sweep-bounded-by-alpha-2}
  \end{align}
\end{lem}

This lemma indicates that if the PPR $\pr_\alpha(\bm{s})$ is far from the stationary distribution $\bm{\pi}_V$ on some vertex set $S$ with a small volume, then there is a sweep cut induced by 
$\pr_\alpha(\bm{s})$ with small volume and conductance.
Thus, we can reduce the clustering problem to showing 
the existence of a vertex set of a large PPR mass.

\section{Local Clustering}\label{sec:local-clustering}

For a hypergraph $H=(V,E,w)$, 
$v \in V$, and $\mu\in (0,1/2]$, we define
\begin{align}
 \phi_{H,v}^\mu=
 \min_{\emptyset \subsetneq C \subsetneq V, v \in C, \atop
 \vol(C)\leq \mu \cdot \vol(V)}  \phi_H(C).
 \label{eq:local-clustering-phi}
\end{align}
In the \emph{local clustering problem}, given a hypergraph $H=(V,E,w)$ and a vertex $v \in V$,
the goal is to find a vertex set containing $v$ whose volume is bounded by $\mu \cdot \vol(V)$
with a conductance close to $\phi_{H,v}^\mu$.
Let $C_v^*$ be an arbitrary minimizer of~\eqref{eq:local-clustering-phi}.
We define $w_\text{min}=\min_{e\in E}\{w(e)\}$, and $w_\text{max}=\max_{e\in E}\{w(e)\}$. 
In this section, we show the following:
\begin{thm}\label{thm:local-clustering}
  Let $H=(V,E,w)$ be a hypergraph, $v \in V$ be a vertex with $v \in {(C^*_v)}^{\circ}$, and $\mu \in (0,1/2]$.
 Our algorithm (Algorithm~\ref{alg:local-clustering}) returns a vertex set $S \subseteq V$ with $v \in S$, $\vol(S)\leq \mu \cdot \vol(V)  + \max_{v\in V}\{d_v\}
  $, and
  $\phi_H(S) = O\left(\sqrt{\phi_{H,v}^\mu}\right)$
  in $O\left((T_H + \sum_{e\in E}|e|) \log (w_\mathrm{max}\sum_{e\in E}|e|/w_\mathrm{min})\right)$ time,
where $T_H$ is the time to compute PPR on $H$.
\end{thm}
We stress here again that Theorem~\ref{thm:local-clustering} is the first local clustering algorithm with a theoretical guarantee on the conductance.
If we use the Euler method to compute PPR, then we have $T_H = O(\sum_{e\in E}|e|/\Delta)$, where $\Delta$ is the time span for one step.
Hence, the total time complexity is $O(\sum_{e\in E}|e|\log(w_\mathrm{max}\sum_{e\in E}|e|/w_\mathrm{min}) /\Delta)$.
We also remark that we can relax the assumption $v\in (C^*_v)^\circ$,
if we admit using the values of the PPR in the assumption.

\subsection{Algorithm}

\begin{algorithm}[t!]
  \caption{Local clustering}\label{alg:local-clustering}
  \Procedure{{\emph{\Call{LocalClustering}{$H,v,\mu$}}}}{
    \Input{\ A hypergraph $H=(V,E,w)$, a vertex $v \in V$, and $\mu \in (0,1/2]$}
    \Output{\ A vertex set $S\subseteq V$}
    $A_{\mathrm{cand}} \leftarrow \left\{ \frac{w_\text{min}(1+\epsilon)^i}{w_\text{max}\sum_{e\in E}|e|} \mid i \in \mathbb{Z}_+ \right\}$ for a constant $\epsilon \in (0,1)$\;
    \textbf{for} each \emph{$\alpha \in A_{\mathrm{cand}}$} \textbf{do} Compute $\textbf{pr}_\alpha(\bm{\chi}_v)$\;

    \Return $\argmin\{\phi_H(S) : S \in \bigcup_{\alpha \in A_{\mathrm{cand}}}\sweep^\mu(\textbf{pr}_\alpha(\bm{\chi}_v))\}$\label{line:sweep}.
  }
\end{algorithm}

We describe our algorithm for local clustering.
Given a hypergraph $H=(V,E,w)$ and a vertex $v \in V$, we first construct a set $A_{\mathrm{cand}}=\left\{ \frac{w_\text{min}(1+\epsilon)^i}{w_\text{max}\sum_{e\in E}|e|} \mid i \in \mathbb{Z}_+ \right\} \cap [0,1]$ for a constant $\epsilon \in (0,1)$. 
As the minimum conductance of a set in a connected hypergraph is lower-bounded by $\frac{w_\text{min}}{w_\text{max}\sum_{e \in E}|e|}$, we can guarantee that there exists $\alpha \in A_{\mathrm{cand}}$ such that
  $\alpha \leq \phi_{H,v}^\mu \leq (1+\epsilon)\alpha$.
Then for each $\alpha \in A_{\mathrm{cand}}$, we compute the PPR $\textbf{pr}_\alpha(\bm{\chi}_v)$.
Finally, we return the set $S$ with the minimum conductance in 
\[
\bigcup_{\alpha \in A_{\mathrm{cand}}}\sweep^\mu(\textbf{pr}_\alpha(\bm{\chi}_v)),
\]
where $\sweep^\mu(\bm{x})=\{S^{\bm{x}}_1,\dots, S^{\bm{x}}_{\ell_\mu}\}$. 
Proposition~\ref{prop:bounds-of-values-of-PPR} implies that
the normalized PPR value $\pr_\alpha(\bm{\chi}_v)(v)/d_v$ takes the
maximum value at $v$; hence
the returned set $S$ includes the specified vertex $v$.
The pseudocode of our local clustering algorithm is given in Algorithm~\ref{alg:local-clustering}.


\subsection{Leak of personalized PageRank}
As we want to apply
Lemma~\ref{lem:upperbound-of-sweep-conductance-local-2} on $C_v^*$, we need to bound the amount of the leak of PPR from $C_v^*$.
To this end, we start by showing the following upper bound.
\begin{lem}\label{lem:ppr-cuts-inequality}
  Let $C \subseteq V$ be a vertex set, $w \in C$, and $H_{\bm{p}} = (V, E_{\bm{p}}, w_{\bm{p}})$ be the graph defined for $\bm{p} = \bm{\pr}_\alpha(\bm{\chi}_w)$ as shown in Section~\ref{subsec:pagerank-hypergraph}.
  Then for any $\alpha \in (0,1]$, we have
  \[
    \bm{p}(\ol{C}) \leq
    \frac{1-\alpha}{2\alpha} \sum_{u \in C} \sum_{v \in \ol{C}} \frac{w_{\bm{p}}(uv)}{d_u}\bm{p}(u),
  \]
\end{lem}
\begin{proof}
  By the definition of $H_{\bm{p}}$, the PPR $\bm{p}$ satisfies
  \begin{align*}
    \bm{p}(\ol{C}) & = \alpha \bm{\chi}_w(\ol{C}) +
    (1-\alpha) W_{H_{\bm{p}}} (\bm{p}) (\ol{C}) \\
    & = (1-\alpha) \bm{p}(\ol{C}) +  (1-\alpha) \frac{1}{2}\left(A_{H_{\bm{p}}} D_{H_{\bm{p}}}^{-1} - I\right) \bm{p}(\ol{C}).
  \end{align*}
  Here, the second equality follows from $\bm{\chi}_w(\ol{C})=0$ and $W_{H_{\bm{p}}} = I + \left(A_{H_{\bm{p}}} D_{H_{\bm{p}}}^{-1} - I\right)/2$.
  Then, we have
  \begin{align}
   \bm{p}(\ol{C}) = \frac{1-\alpha}{2\alpha}
   \left(A_{H_{\bm{p}}} D_{H_{\bm{p}}}^{-1} - I\right)\bm{p}(\ol{C}).\label{eq:pagerank-equality-C-bar}
  \end{align}
For a subset $S\subseteq V$, we define two sets of
ordered pairs of vertices: $\mathrm{in}(S) = \{ (u,v) \in V\times S \}$ and 
$\mathrm{out}(S) = \{ (u,v) \in S \times V \}$. 
For an ordered pair $(u,v) \in V\times V$,
we set $\bm{p}(u,v) := \frac{w_{\bm{p}}(uv)}{d_u}\bm{p}(u)$.
Then, the term $\left(A_{H_{\bm{p}}} D_{H_{\bm{p}}}^{-1} - I\right)(\bm{p}) (\ol{C})$ can be rewritten as
  \begin{align}
   & \left(A_{H_{\bm{p}}} D_{H_{\bm{p}}}^{-1} - I\right)(\bm{p}) (\ol{C}) \nonumber
   = \sum_{u \in \ol{C}}\sum_{v \in V} \frac{w_{\bm{p}}(uv)}{d_{v}} \bm{p}(v) - \sum_{u \in \ol{C}} \bm{p}(u) \nonumber \\
   &= \sum_{u \in \ol{C}}\sum_{v \in V} \frac{w_{\bm{p}}(uv)}{d_{v}} \bm{p}(v) - \sum_{u \in \ol{C}} \sum_{v\in V} \frac{w_{\bm{p}}(uv)}{d_{u}}\bm{p}(u) \nonumber \\
    & = \sum_{(u,v) \in \mathrm{in}(\ol{C}) \! \setminus \!
    \mathrm{out}(\ol{C})} \bm{p}(u,v) -
    \sum_{(u,v) \in\mathrm{out}(\ol{C})\!\setminus\!\mathrm{in}(\ol{C})} \bm{p}(u,v) \nonumber \\
   & \leq \sum_{(u,v) \in \mathrm{in}(\ol{C})\!\setminus\! \mathrm{out}(\ol{C})} \bm{p}(u,v)
   = \sum_{u \in C} \sum_{v \in\ol{C}} \frac{w_{\bm{p}}(uv)}{d_u}\bm{p}(u).
    \label{eq:pagerank-equality-C-bar-3}
  \end{align}

  The claim follows by combining~\eqref{eq:pagerank-equality-C-bar} 
   and~\eqref{eq:pagerank-equality-C-bar-3}.
\end{proof}

Next, we show that the amount of the leak of PPR from a vertex set $C$ can be bounded using its conductance.
\begin{lem}\label{lem:local-bound-of-PR}
  Let $C \subseteq V$ be a vertex set with $\vol(C)\leq \vol(V)/2$.
  Then, for any $\alpha \in (0,1]$ and a vertex $v \in {C}^\circ$, we have
  \[
    \bm{\pr}_\alpha(\bm{\chi}_v)(\ol{C}) \leq \frac{1}{4\alpha} \phi_H(C).
  \]
\end{lem}
\begin{proof}
  Let $\bm{p} = \bm{\pr}_\alpha(\bm{\chi}_v)$.
  By Lemma~\ref{lem:ppr-cuts-inequality}, we have
  \begin{align}
   \frac{2\alpha}{1-\alpha}\bm{p}(\ol{C}) \leq
   \sum_{u' \in C} \sum_{v' \in\ol{C}} \frac{w_{\bm{p}}(u'v')}{d_{u'}}\bm{p}(u').\label{eq:inequality-of-PPRs}
  \end{align}
  As $v \in C$, by Proposition~\ref{prop:bounds-of-values-of-PPR},
  the RHS of~\eqref{eq:inequality-of-PPRs} is at most
  \[
   \sum_{u' \in C} \sum_{v' \in\ol{C}} \frac{w_{\bm{p}}(u'v')}{d_{u'}}
   \frac{d_{u'}}{\vol(V)} +
   \sum_{v' \in\ol{C}} \frac{w_{\bm{p}}(vv')}{d_{v}}\bm{p}(v)
   - \sum_{v' \in\ol{C}} \frac{w_{\bm{p}}(vv')}{d_{v}}
   \frac{d_{v}}{\vol(V)}.
  \]
  From the assumption $v \in C^{\circ}$, $w_{\bm{p}}(vv')$ is zero for any
  $v' \in \ol{C}$. Hence, we have
  \begin{align*}
  \bm{p}(\ol{C})&\leq \frac{1-\alpha}{2\alpha}
  \sum_{u' \in C} \sum_{v' \in\ol{C}} \frac{w_{\bm{p}}(u'v')}{\vol(V)}
  \leq \frac{1-\alpha}{2\alpha}\frac{\vol(C)}{\vol(V)} \phi_H(C)\\
  & \leq \frac{1-\alpha}{4\alpha} \phi_H(C)
  \leq \frac{1}{4\alpha} \phi_H(C) .
  \end{align*}
The second inequality follows from
\begin{align}
 \sum_{u' \in C} \sum_{v' \in\ol{C}} w_{\bm{p}}(u'v')
 \leq \sum_{e\in \partial C} w(e).\label{eq:cut-between-graph-hypergraph} 
\end{align}
\end{proof}
\subsection{Proof of Theorem~\ref{thm:local-clustering}}\label{subsec:theorem-local}


First, we discuss the conductance of the output vertex set.
\begin{lem}\label{lem:main-local}
Let $\mu\in (0,1/2]$.
Suppose $v \in {(C^*_v)}^\circ$.
  Then for any $\epsilon \in (0,1)$ and $\alpha \in (0,1]$ with $\alpha \leq \phi_H\left(C_v^*\right) \leq (1+\epsilon)\alpha$,
  we have
  \[
    \phi_H^\mu(\bm{\pr}_\alpha(\bm{\chi}_v)) < \frac{8}{
    \sqrt{3-\epsilon-4\mu}}\sqrt{6\phi_H\left(C_v^*\right) \log\left( \frac{2}{3-\epsilon -4\mu}\right)}.
  \]
\end{lem}
\begin{proof}
Let $C = C_v^*$ and $ \bm{p} = \bm{\pr}_\alpha(\bm{\chi}_v)$.
To apply Lemma~\ref{lem:upperbound-of-sweep-conductance-local-2}, we want to lower-bound $\bm{p}(C) - \bm{\pi}_V(C)$.
By the fact that the PPR $\bm{p}$ is a distribution by Theorem~\ref{thm:existence-of-PR} and Lemma~\ref{lem:local-bound-of-PR} for $\alpha \geq \phi_H(C)/(1 + \epsilon)$,
we have $\bm{p}(C) -\bm{\pi}_V(C) = 1 -
 \bm{p}(\ol{C}) - \bm{\pi}_V(C) \geq 1 - (1+\epsilon)/4 - \mu =
 (3-\epsilon-4\mu)/4$.
Applying Lemma~\ref{lem:upperbound-of-sweep-conductance-local-2} with $C = C_v^*$,
$\delta = (3-\epsilon-4\mu)/4$, and $\alpha \leq \phi_H(C)$, we obtain
\begin{align*}
 \phi_H^\mu(\bm{p})
  < \frac{8}{\sqrt{3-\epsilon-4\mu}}\sqrt{6\phi^\mu_{H,v} \log\left(\frac{2}{3-\epsilon -4\mu} \right)}. &\qedhere
\end{align*}
\end{proof}

\begin{proof}[Proof of Theorem~\ref{thm:local-clustering}]
  The guarantee on conductance holds by Lemma~\ref{lem:main-local} as some $\alpha \in A_{\mathrm{cand}}$ satisfies the condition of Lemma~\ref{lem:main-local} with $\epsilon \in (0,1)$.

  Now, we analyze the time complexity.
  As the size of $|A_{\mathrm{cand}}| = O(\log (w_\text{max}\sum_{e\in E}|e|/w_\text{min}))$, it takes $O(T_H \log (w_\text{max}\sum_{e\in E}|e|/w_\text{min}))$ time to compute all the PPRs.
  Computing the conductance values of the sweep cuts induced by a particular vector takes $O(\sum_{e \in E}|e|)$ time.
  Hence, Line~\ref{line:sweep} takes $O(\sum_{e\in E}|e|\log (w_\text{max}\sum_{e\in E}|e|/w_\text{min}))$ time.
  Combining those, we obtain the claimed time complexity.
\end{proof}



\section{Global Clustering}\label{sec:global-clustering}
In the \emph{global clustering problem}, given a hypergraph $H$, the goal is to find a vertex set with conductance close to $\phi_H$.
In this section, we show the following: 
\begin{thm}\label{thm:global-clustering}
  Let $H=(V,E,w)$ be a hypergraph.
  Our algorithm (Algorithm~\ref{alg:global-clustering}) 
  returns a vertex set $S \subseteq V$ with 
  $\phi_H(S) = O\left(\sqrt{\phi_{H}}\right)$ 
  in $O\left((T_H + \sum_{e\in E}|e|)n\log (w_\mathrm{max}\sum_{e\in E}|e|/w_\mathrm{min})\right)$ time, where $T_H$ is the time to compute PPR on $H$.
\end{thm}

\subsection{Algorithm}

\begin{algorithm}[t!]
  \caption{Global clustering}\label{alg:global-clustering}
  \Procedure{{\emph{\Call{GlobalClustering}{$H$}}}}{
    \Input{\ A hypergraph $H=(V,E,w)$}
    \Output{\ A vertex set $S\subseteq V$}
    \textbf{for} $v \in V$ \textbf{do} $S_v \leftarrow \Call{LocalClustering}{H,v,1/2}$\label{line:global_V}\;
    \Return $\argmin\{\phi_H(S):S\in \{S_v:v\in V\}\}$. 
  }
\end{algorithm}

Our global clustering algorithm is quite simple.
That is, given a hypergraph $H=(V,E,w)$, we simply call \Call{LocalClustering}{$H,v,1/2$} for every $v \in V$, and then return the best returned set in terms of conductance.
The pseudocode is given in Algorithm~\ref{alg:global-clustering}. 

\subsection{Leak of personalized PageRank}
Let $C^*$ be an arbitrary set of conductance $\phi_H$.
As we want to apply
Lemma~\ref{lem:upperbound-of-sweep-conductance-local-2} for $\mu = 1/2$
on $\phi_H$, we need to bound the amount of the leak of PRR
from $\phi_H$.
The following lemma is a counterpart to Lemma~\ref{lem:local-bound-of-PR} for the global setting.
\begin{lem}\label{lem:bound-by-conductance}
  Suppose that $\alpha \in (0,1]$ and $C \subseteq V$ satisfy
  \begin{align}
    \left( \sum_{w \in C}\bm{\pi}_C(w) \bm{\pr}_\alpha(\bm{\chi}_w)\right)(v)\leq \bm{\pi}_C(v)\label{eq:assumption}
  \end{align}
  for every $v \in C\!\setminus\!C^\circ$.
  Then we have
  \[
    \sum_{w\in C} \bm{\pi}_C(w)\bm{\pr}_\alpha(\bm{\chi}_w)(\ol{C}) \leq \frac{\phi_H(C)}{2\alpha}.
  \]
\end{lem}
The assumption~\eqref{eq:assumption} means that the expected PPR of a boundary vertex $v \in C\!\setminus\!C^\circ$ is at most its probability mass in the distribution $\bm{\pi}_C$.
We note that~\eqref{eq:assumption} always holds as an equality when $H$ is a graph.
We discuss two sufficient conditions of this assumption in Lemmas~\ref{lem:criterion-1} and~\ref{lem:criterion-2} in Section~\ref{subsec:sufficient-conditions}.
These conditions show that the assumption of Lemma~\ref{lem:bound-by-conductance} is quite mild. 
\begin{proof}[Proof of Lemma~\ref{lem:bound-by-conductance}]
We set $\bm{p}_{w} = \bm{\pr}_\alpha(\bm{\chi}_w)$ and let $w_0 \in C$ be an arbitrary vertex that maximizes $\sum_{v \in \ol{C}} w_{\bm{p}_w}(uv)$.
Then, by Lemma~\ref{lem:ppr-cuts-inequality}, we have
\begin{align}
  \left(A_H D_H^{-1} - I\right)(\bm{p}_w) (\ol{C})\leq \sum_{u \in C} \sum_{v \in\ol{C}} \frac{w_{\bm{p}_{w_0}}(uv)}{d_u}\bm{p}_w(u).
\label{eq:ppr-cuts-inequality}
\end{align}
By taking the average of the inequality~\eqref{eq:ppr-cuts-inequality} using the distribution $\bm{\pi}_C$,
we obtain the following inequality:
\begin{align}
 & \sum_{w\in C}\bm{\pi}_C(w) \left(A_H D_H^{-1} - I\right)(\bm{p}_w) (\ol{C}) \nonumber \\
 & \leq \sum_{u \in C} \sum_{v \in\ol{C}} \frac{w_{\bm{p}_{w_0}}(uv)}{d_u}
 \left( \sum_{w \in C}\bm{\pi}_C(w)
 \bm{p}_w\right) (u).\label{eq:sum-inequality}
\end{align}
Now, we have
\begin{align*}
 & \sum_{w\in C} \bm{\pi}_C(w)\bm{p}_w(\ol{C})
 =  \sum_{w\in C} \bm{\pi}_C(w)\left(\frac{1-\alpha}{2\alpha}
 \left(A_H D_H^{-1} - I\right)\bm{p}_w(\ol{C})\right) \tag{by~\eqref{eq:pagerank-equality-C-bar}}\\
 &\leq \frac{1-\alpha}{2\alpha}
 \sum_{u \in C} \sum_{v \in\ol{C}} \frac{w_{\bm{p}_{w_0}}(uv)}{d_u}
 \left( \sum_{w \in C}\bm{\pi}_C(w) \bm{p}_w\right) (u) \tag{by~\eqref{eq:sum-inequality} }\\
 & \leq \frac{1-\alpha}{2\alpha}
 \sum_{u \in C} \sum_{v \in\ol{C}} \frac{w_{\bm{p}_{w_0}}(uv)}{d_u} \frac{d_u}{\vol(C)} \tag{by the assumption} \\
 & \leq \frac{1-\alpha}{2\alpha} \phi_H(C) \leq \frac{1}{2\alpha} \phi_H(C). \tag{by~\eqref{eq:cut-between-graph-hypergraph}} 
\end{align*}
\end{proof}


\subsection{Proof of Theorem~\ref{thm:global-clustering}}\label{subsec:theorem-global}
For a subset $C \subseteq V$ and $\alpha \in (0,1]$, we define a subset
$C_\alpha \subseteq C$ as
\[
 C_\alpha = \left\{ v \in C \ \left| \ \bm{\pr}_\alpha(\bm{\chi}_v)(\ol{C}) \leq \frac{\phi_H(C)}{\alpha} \right. \right\}.
\]
We here show the following from which Theorem~\ref{thm:global-clustering} easily follows.
\begin{thm}\label{thm:main}
  Let $C^* \subseteq V$ be a set with $\phi_H = \phi_H(C^*)$ and assume
  $\vol(C^*) \leq \vol(V)/2$.
  Suppose that $\alpha \leq 10 \phi_H \leq (1+\epsilon)\alpha$ and
  the condition~\eqref{eq:assumption} hold.
  Then, for any $v\in {(C^*)}_\alpha$,
  we have
  \[
    \phi_H^{1/2}(\bm{\pr}_\alpha(\bm{\chi}_v)) < \frac{20}{\sqrt{4-\epsilon}}\sqrt{3\phi_H \log\left( \frac{40}{4-\epsilon}  \right)}.
  \]
\end{thm}

Theorem~\ref{thm:main} follows from Lemma~\ref{lem:upperbound-of-sweep-conductance-local-2} for $\mu = 1/2$ and the following lemma. The proof of Lemma~\ref{lem:lowerbound-of-pagerank} is 
similar to the proof of Theorem~5.1 in \cite{Andersen:2007tsa} 
 with Lemma~10. The detail is in 
Appendix~\ref{sec:proof-lowerbound-of-pagerank}.  
\begin{lem}\label{lem:lowerbound-of-pagerank}
  Suppose a set  $C \subseteq V$ and $\alpha > 0$ satisfy the condition~\eqref{eq:assumption}.
  Then, $\vol(C_\alpha) \geq \vol(C)/2$ holds.
\end{lem}

\begin{proof}[Proof of Theorem~\ref{thm:main}]
By instantiating Lemma~\ref{lem:lowerbound-of-pagerank} with $C = C^*$ and $\alpha \leq 10\phi_H \leq (1+\epsilon)\alpha$,
we have
\[
 \bm{\pr}_\alpha(\bm{\chi}_v)(C^*) \geq 1- \frac{1+\epsilon}{10} =
 \frac{9-\epsilon}{10}
\]
for $v \in (C^*)_\alpha$. Because
$\vol(C^*)\leq \vol(V)/2$,
\[
 \bm{\pr}_\alpha(\bm{\chi}_v)(C^*) - \bm{\pi}(C^*) \geq
 \frac{9-\epsilon}{10} - \frac{1}{2} = \frac{4-\epsilon}{10}.
\]
Instantiating
Lemma~\ref{lem:upperbound-of-sweep-conductance-local-2} with
$\mu = 1/2$, $\delta = (4-\epsilon)/10$, and $\alpha \leq 10\phi_H$, we obtain 
\begin{align*}
 \phi_H^{1/2}(\bm{\pr}_\alpha(\bm{\chi}_v))
 <  \frac{20}{\sqrt{4-\epsilon}}\sqrt{ 3\phi_H \log\left( 
 \frac{40}{4-\epsilon}\right)}.  &\qedhere
\end{align*}
\end{proof}
\begin{proof}[Proof of Theorem~\ref{thm:global-clustering}]
  The bound on conductance follows from Theorem~\ref{thm:main}.
  The analysis on the time complexity is trivial.
\end{proof}

\subsection{Sufficient conditions}\label{subsec:sufficient-conditions}

Here we discuss two useful sufficient conditions of the assumption~\eqref{eq:assumption} in Lemma~\ref{lem:bound-by-conductance}. 
We give the proofs in Appendix~\ref{sec:proof-sufficient-conditions}. 
\begin{lem}\label{lem:criterion-1}
Let $\alpha \in (0,1]$ and $C \subseteq V$ be a vertex set with $\vol(C) \leq \vol(V)/2$.
If $\bm{\pr}_\alpha(\bm{\chi}_v)(v) \leq 1/2$ holds for every $v \in V$, then the condition~\eqref{eq:assumption} holds for any $v\in C\!\setminus\!C^\circ$.
\end{lem}


\begin{lem}\label{lem:criterion-2}
Let $d_{\max} = \max_{v\in V}\{d_v \}$. 
If $\alpha \in (0,1]$ satisfies
\[
 \alpha \leq \left(\frac{1}{2} - \frac{d_{\max}}{\vol(V)}\right)
 {\left(1 - \frac{d_{\max}}{\vol(V)}\right)}^{-1},
\]
then the assumption~\eqref{eq:assumption} holds for any $v\in C\!\setminus\! C^\circ$.
\end{lem}


 Lemma~\ref{lem:criterion-2} implies that, if $d_\text{max}/\vol(V) \leq 1/4$ holds, then Lemma~\ref{lem:bound-by-conductance} holds for any $\alpha \leq 1/4$, which is usually the case in practice.


\section{Experimental Evaluation}\label{sec:experiments}

In this section, we evaluate the performance of Algorithms~\ref{alg:local-clustering} and~\ref{alg:global-clustering}
in terms of both the quality of solutions and running time.

\paragraph{Dataset}
Table~\ref{tab:instance} lists real-world hypergraphs on which our experiments were conducted.
Except for DBLP~KDD, all hypergraphs were collected from KONECT\footnote{
\url{http://konect.uni-koblenz.de/}
}.
As these are originally unweighted bipartite graphs, we transformed each of them into a hypergraph (with edge weights) as follows:
We take the left-hand-side vertices in a given bipartite graph as the vertices in the corresponding hypergraph,
and then for each right-hand-side vertex, we add a hyperedge (with weight one) consisting of the neighboring vertices in the bipartite graph.
DBLP~KDD is a coauthorship hypergraph, which we compiled from the DBLP dataset\footnote{
\url{https://dblp.uni-trier.de/}
}.
Specifically, the vertices correspond to the authors of the papers in KDD 1995--2019,
while each hyperedge represents a paper, which contains the vertices corresponding to the authors of the paper.

As almost all hypergraphs generated as above were disconnected, 
we extracted the largest connected component in each hypergraph.
All information in Table~\ref{tab:instance} is about the resulting hypergraphs.

\begin{table}[t]
\centering
\caption{Real-world hypergraphs used in our experiments.}\label{tab:instance}
\scalebox{0.80}{
\begin{tabular}{lrrrr}
\toprule
Name & $n$   & $m$   &$\sum_{v\in V}d_v/n$ &$\sum_{e\in E}|e|/m$\\
\midrule
Graph products &86  &106  &3.57 &2.04\\
Network theory &330 &299 &3.06 &3.04 \\
DBLP KDD &5,590  &2,719  &1.99  &4.00\\
arXiv cond-mat &13,861  &13,571  &3.87  &3.05\\
DBpedia Writers  &54,909  &18,232  &1.80  &5.29\\
YouTube Group Memberships &88,490  &21,974  &3.24  &12.91\\
DBpedia Record Labels  &158,385  &9,827  &1.40  &22.51\\
DBpedia Genre &253,968  &4,934  &1.80  &92.84\\
\bottomrule
\end{tabular}
}
\end{table}

\paragraph{Parameter setting}
We explain the parameter setting of our algorithms. Throughout the experiments, we set $\epsilon =0.9$.
In our implementation, to make the algorithms more scalable, we try to keep the PageRank vector sparse;
that is, we round every PageRank value less than $10^{-5}$ down to zero in each iteration of the Euler method.
As a preliminary experiment, using the two largest hypergraphs,
we observed the effect of the parameters in the Euler method, $\Delta$ (i.e., time span)  and $T$ (i.e., total time), in terms of the solution quality.
Specifically, we selected 50 initial vertices uniformly at random, and for each of which,
we ran Algorithm~\ref{alg:local-clustering} with $\mu=1/2$ for each pair of parameters $(\Delta, T)\in \{0.5, 1.0, 2.0\}\times \{2, 4,\dots, 30\}$.
The results are depicted in Figure~\ref{fig:pre}.
As can be seen, the effect of $\Delta$ is drastic.
The performance with $\Delta=0.5$ is worse than that with $\Delta=1.0$, 
which implies that due to the above-mentioned sparsification of the PageRank vector, 
we have to spend some amount of time at once to obtain a non-negligible diffusion.
Moreover, the performance with $\Delta=2.0$ is not stable, meaning that we should not use an unnecessarily large value for $\Delta$.
The effect of $T$ is also significant; when $\Delta=0.5$ or 1.0, the conductance of the output solution becomes smaller as $T$ increases.
From these observations, in the following main experiments, we set $\Delta=1.0$ and $T=30$.

\begin{figure}[t]
\begin{subfloat}[DBpedia Record Labels]{
\includegraphics[width=0.227\textwidth]{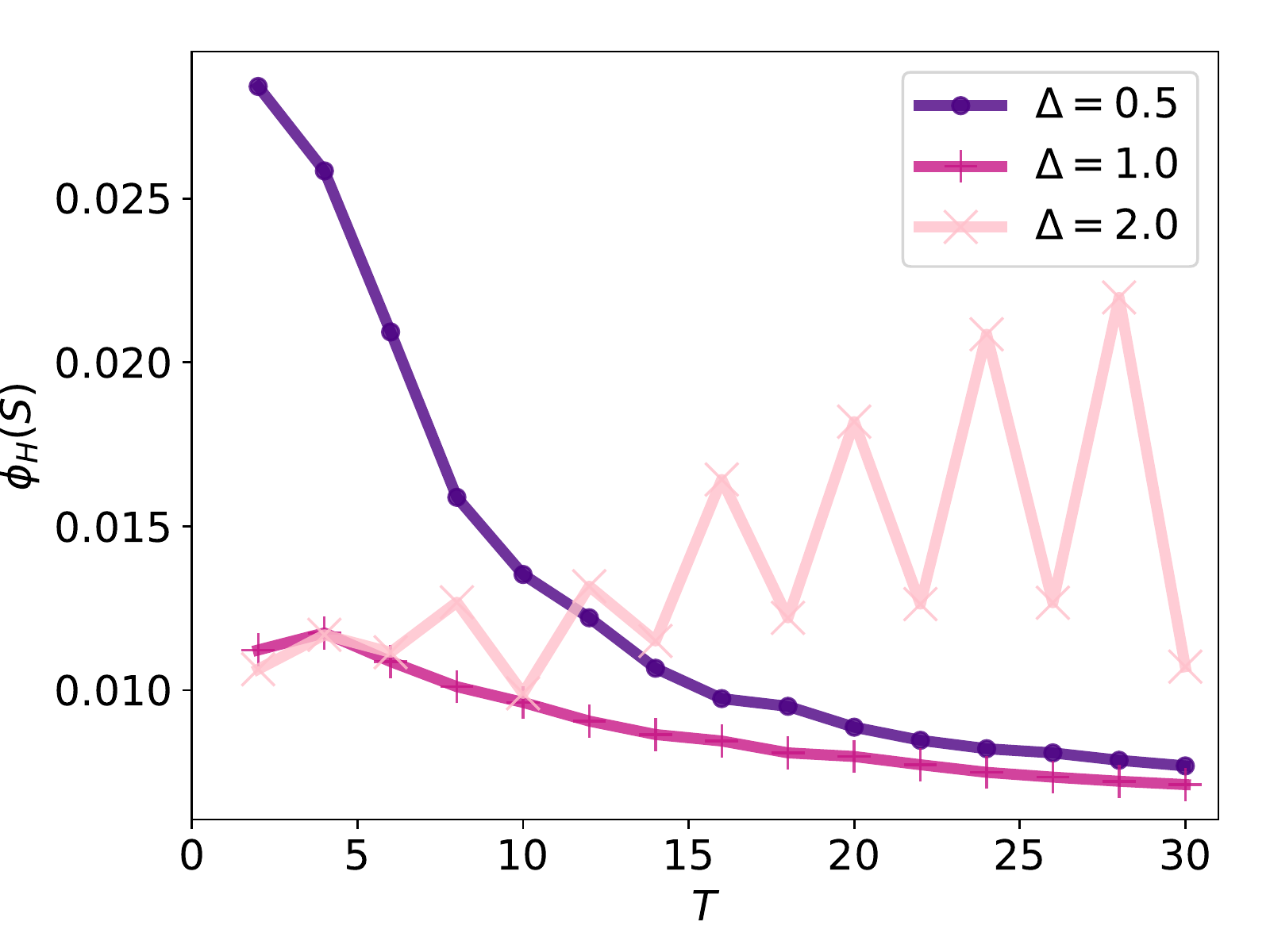}
}
\end{subfloat}
\begin{subfloat}[DBpedia Genre]{
\includegraphics[width=0.227\textwidth]{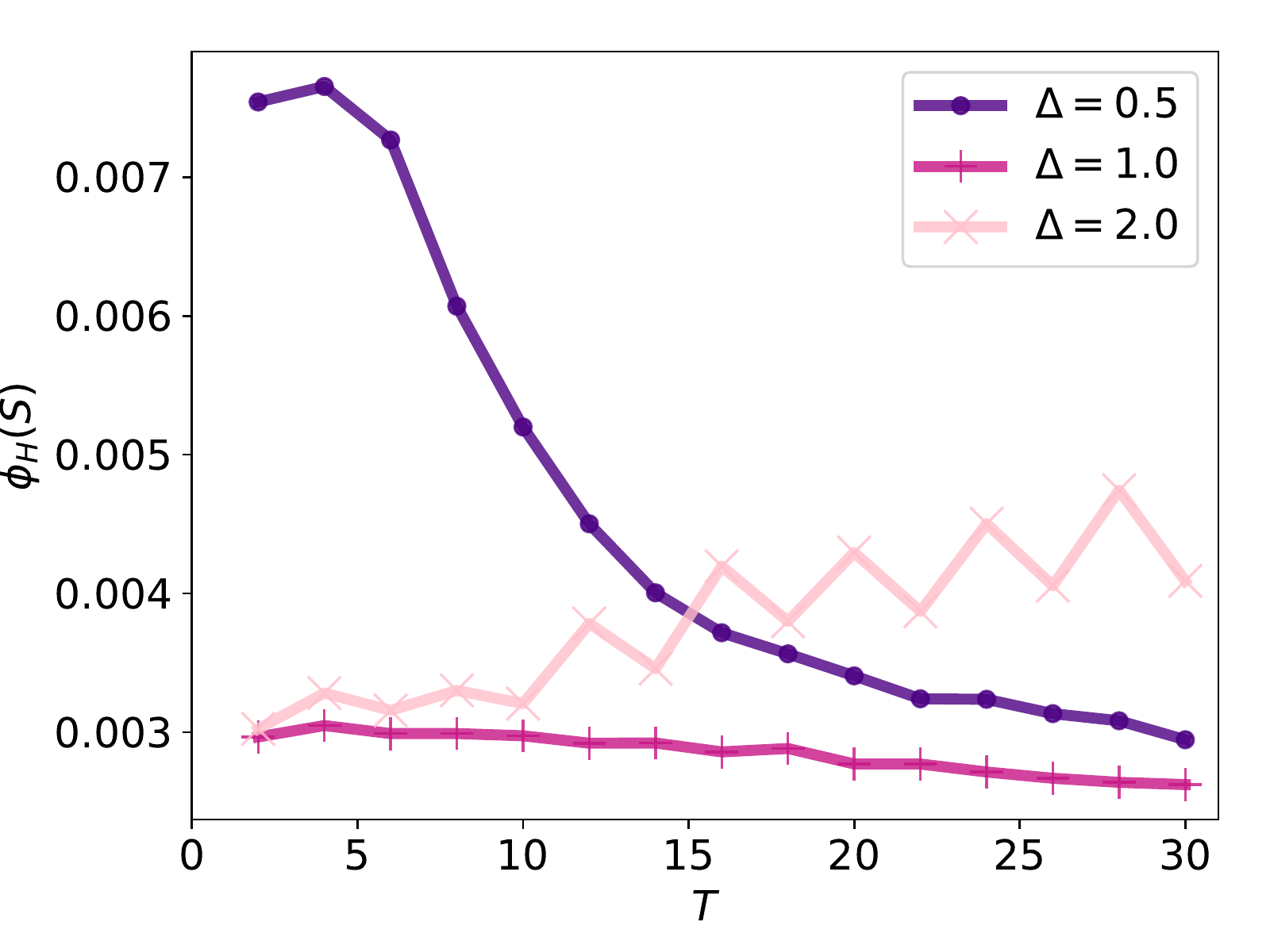}
}
\end{subfloat}
\caption{Effect of parameters $\Delta$ and $T$. Average values over 50 initial vertices are plotted.}\label{fig:pre}
\end{figure}

\begin{figure*}[t]
\begin{subfloat}[Graph products]{
\includegraphics[width=0.237\textwidth]{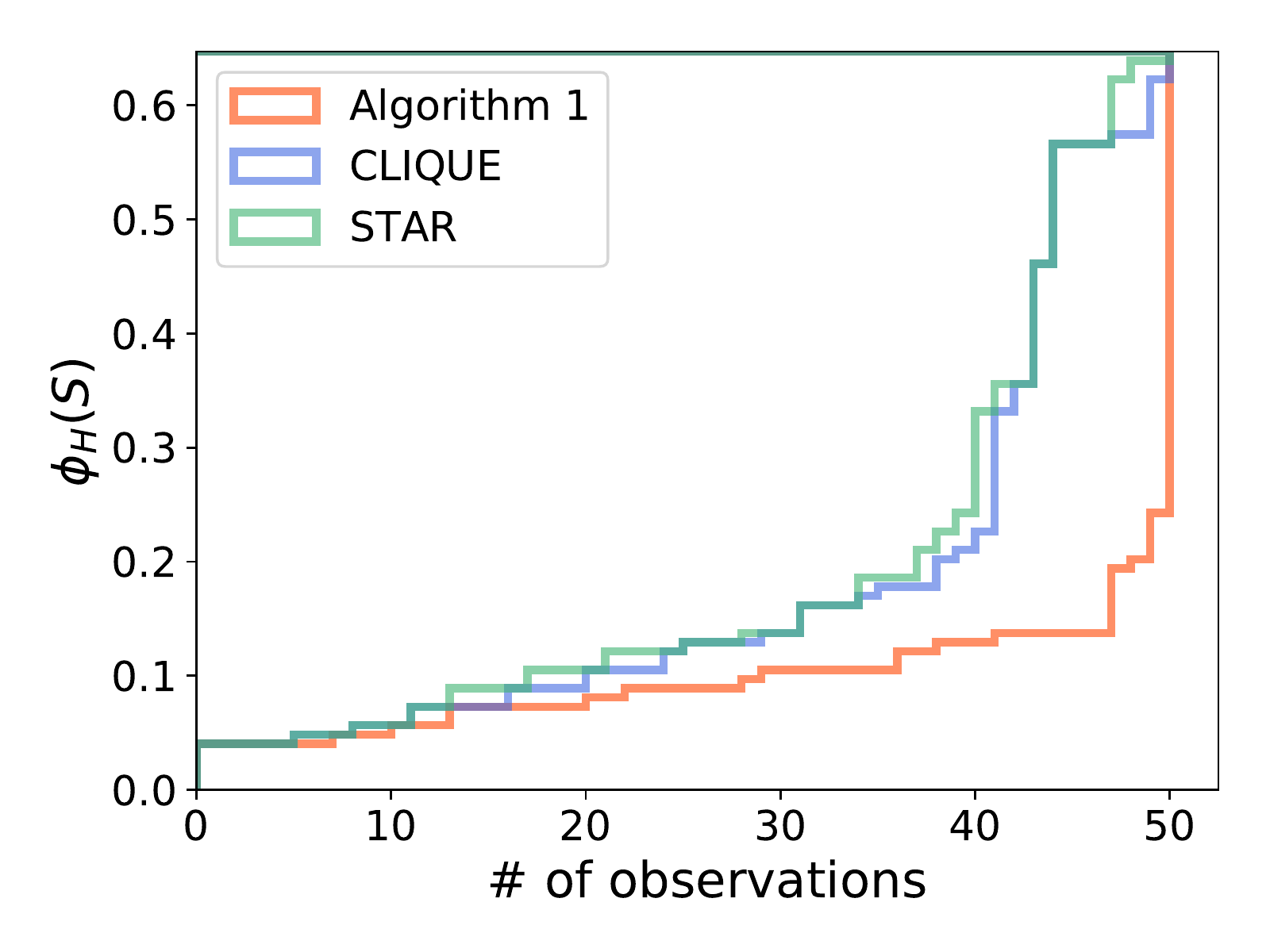}
}
\end{subfloat}
\begin{subfloat}[Network theory]{
\includegraphics[width=0.237\textwidth]{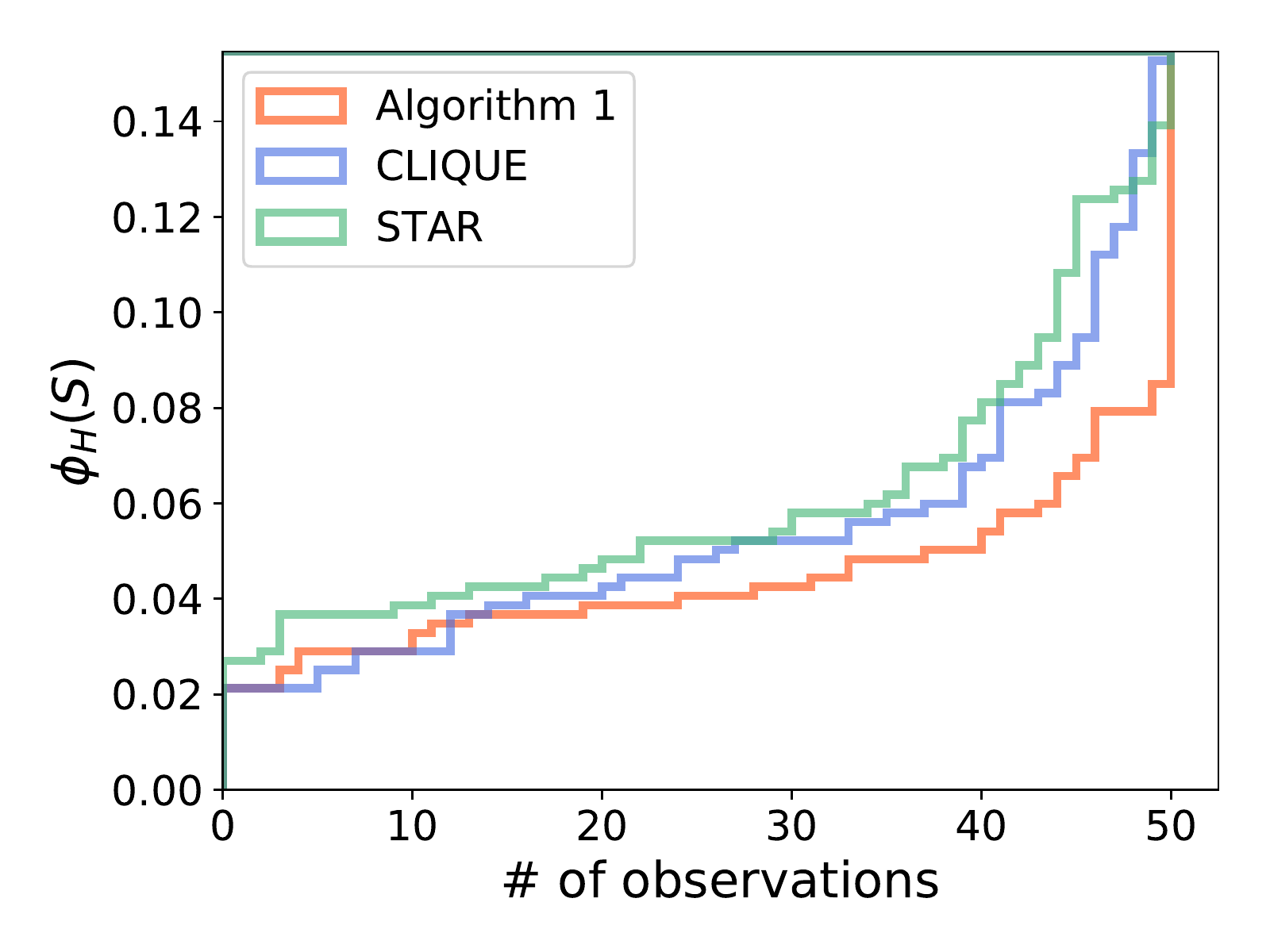}
}
\end{subfloat}
\begin{subfloat}[DBLP KDD]{
\includegraphics[width=0.237\textwidth]{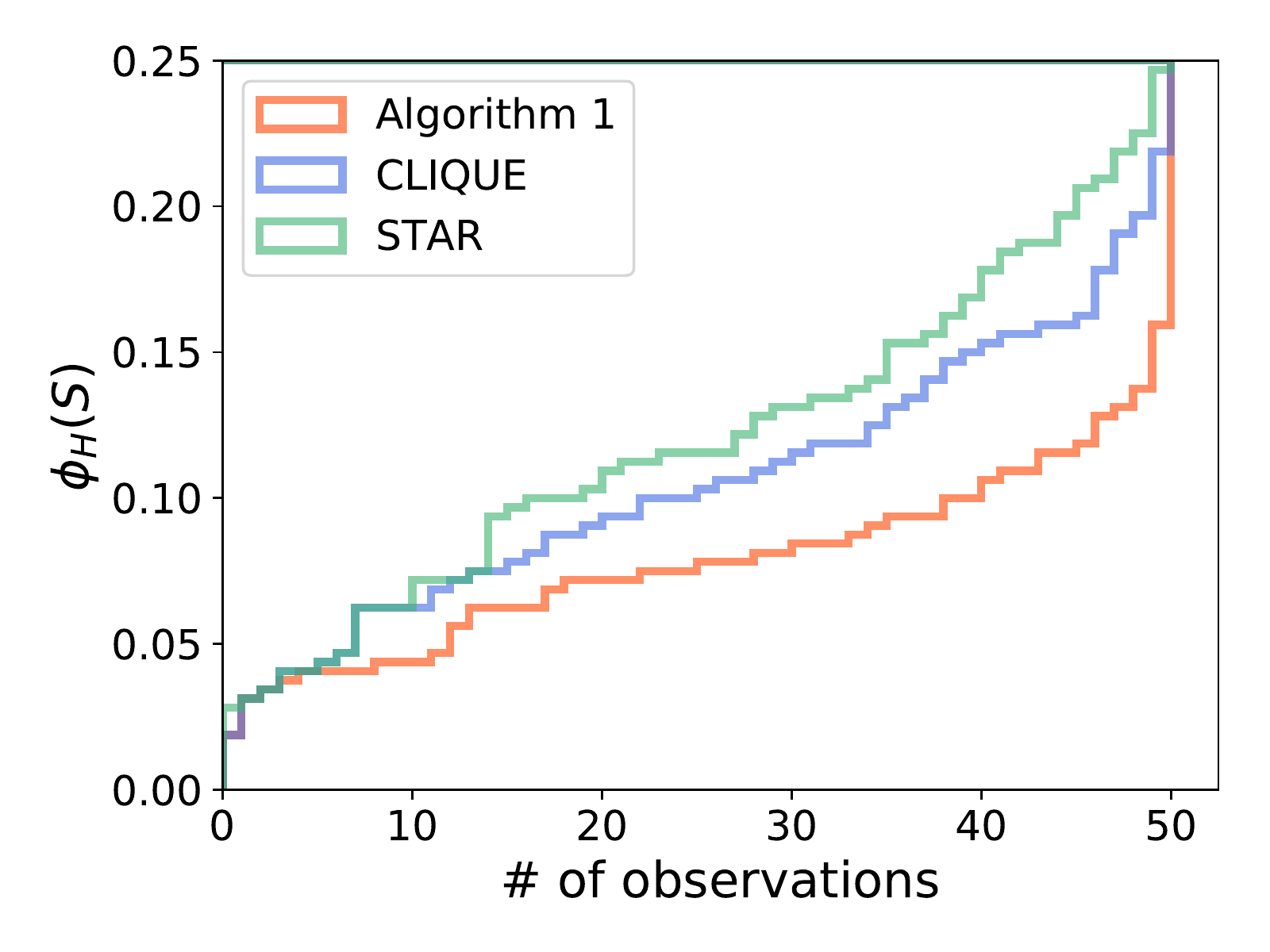}
}
\end{subfloat}
\begin{subfloat}[arXiv cond-mat]{
\includegraphics[width=0.237\textwidth]{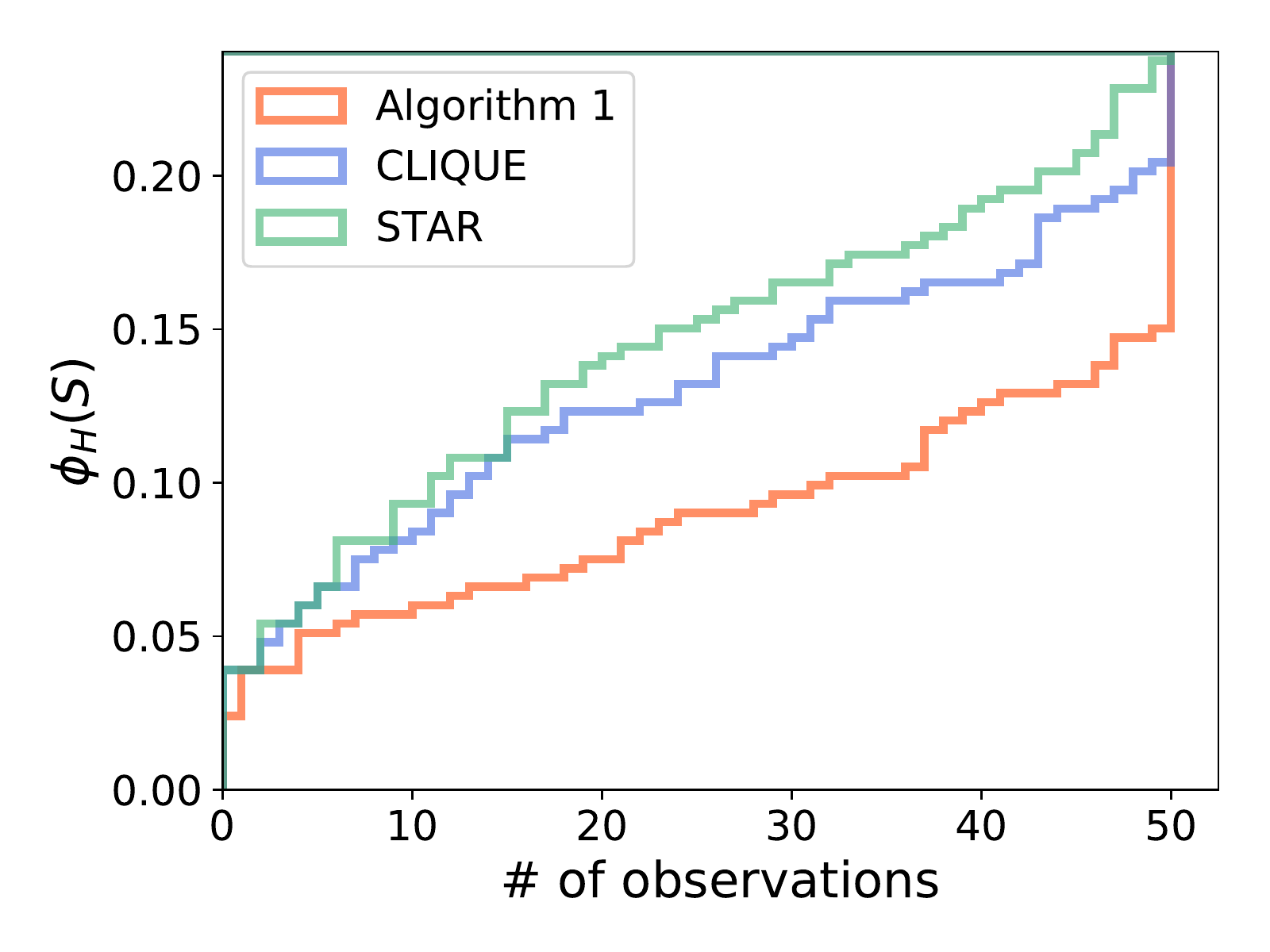}
}
\end{subfloat}\\
\begin{subfloat}[DBpedia Writers]{
\includegraphics[width=0.237\textwidth]{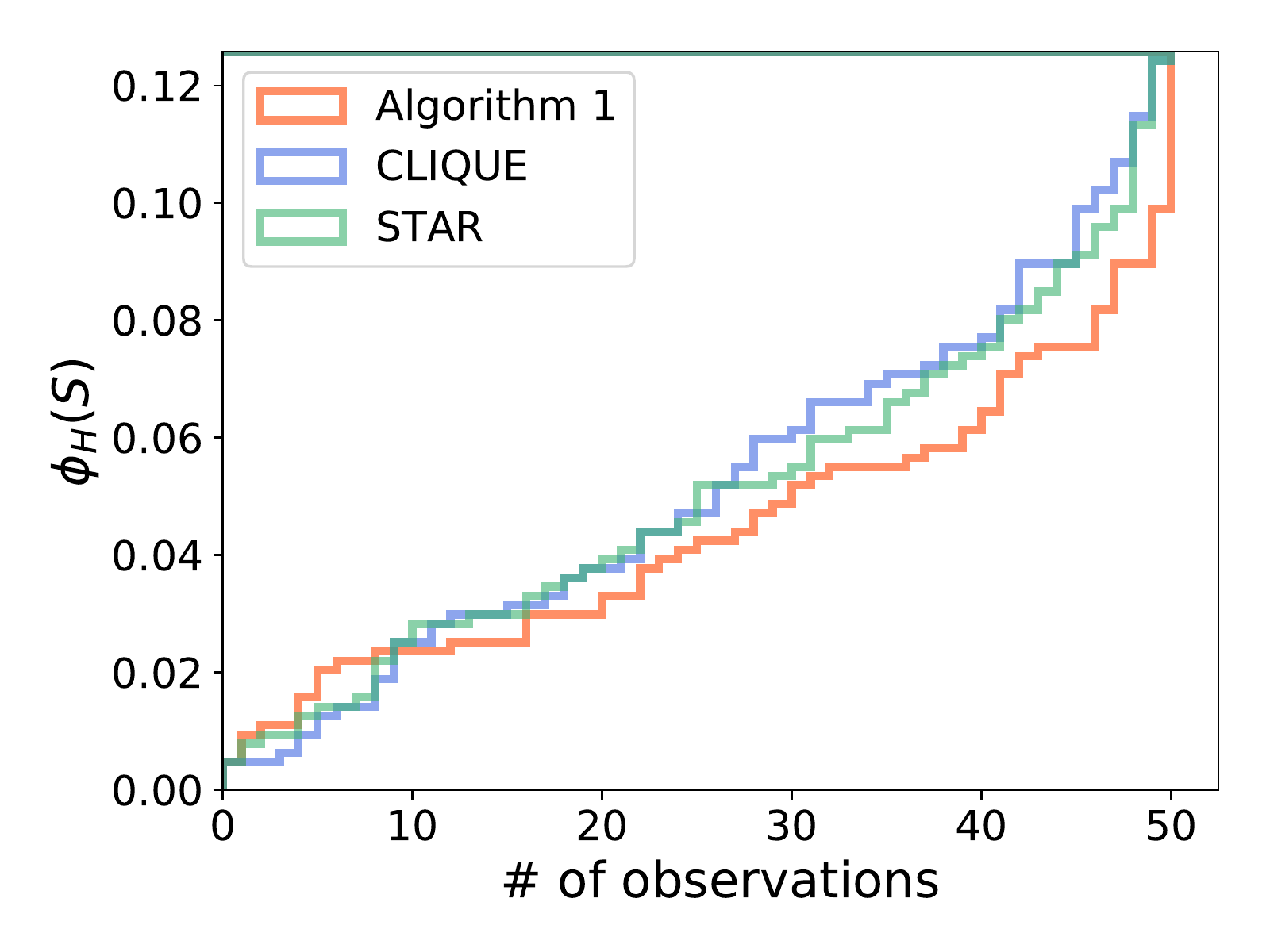}
}
\end{subfloat}
\begin{subfloat}[YouTube Group Memberships]{
\includegraphics[width=0.237\textwidth]{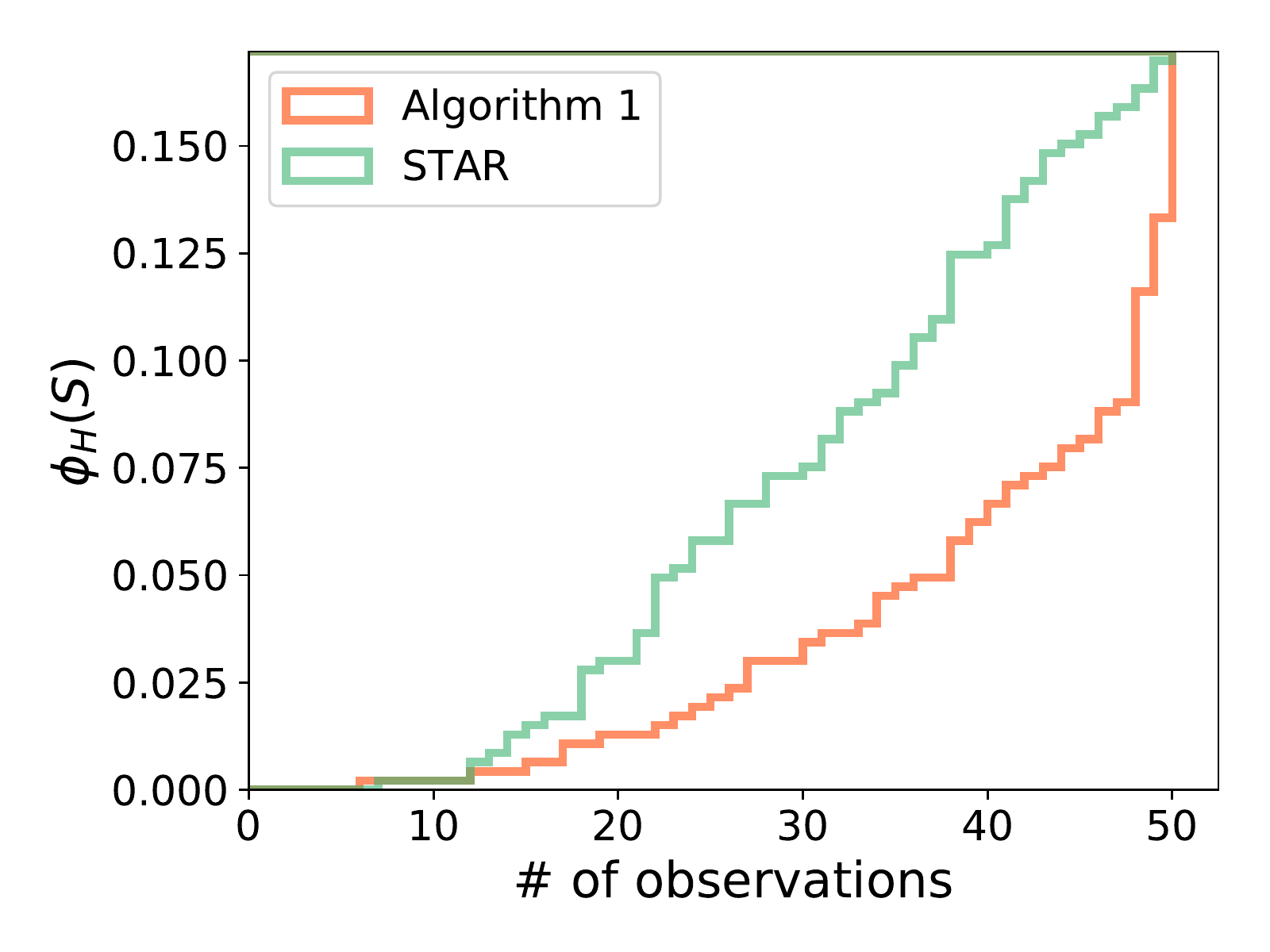}
}
\end{subfloat}
\begin{subfloat}[DBpedia Record Labels]{
\includegraphics[width=0.237\textwidth]{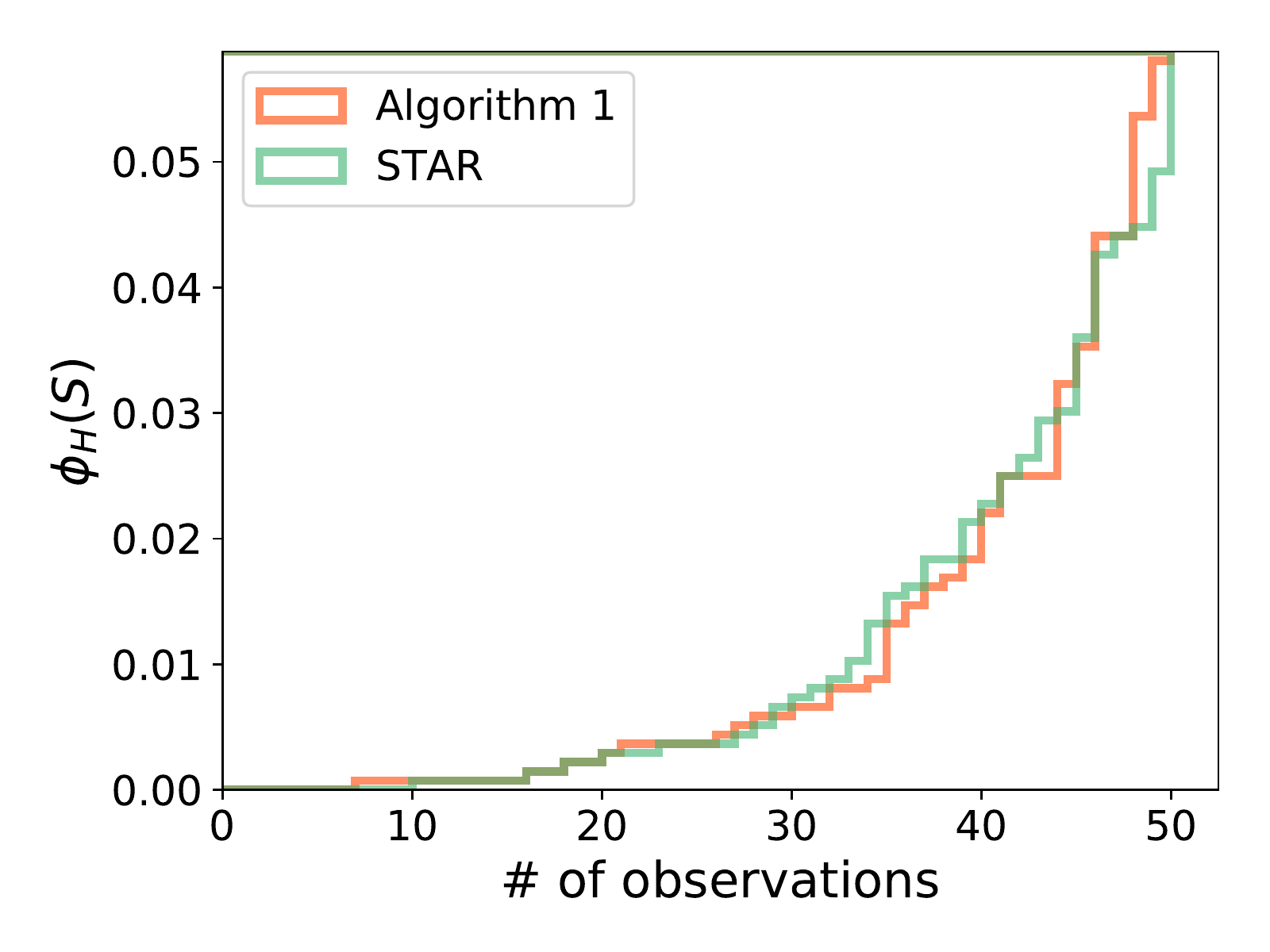}
}
\end{subfloat}
\begin{subfloat}[DBpedia Genre]{
\includegraphics[width=0.237\textwidth]{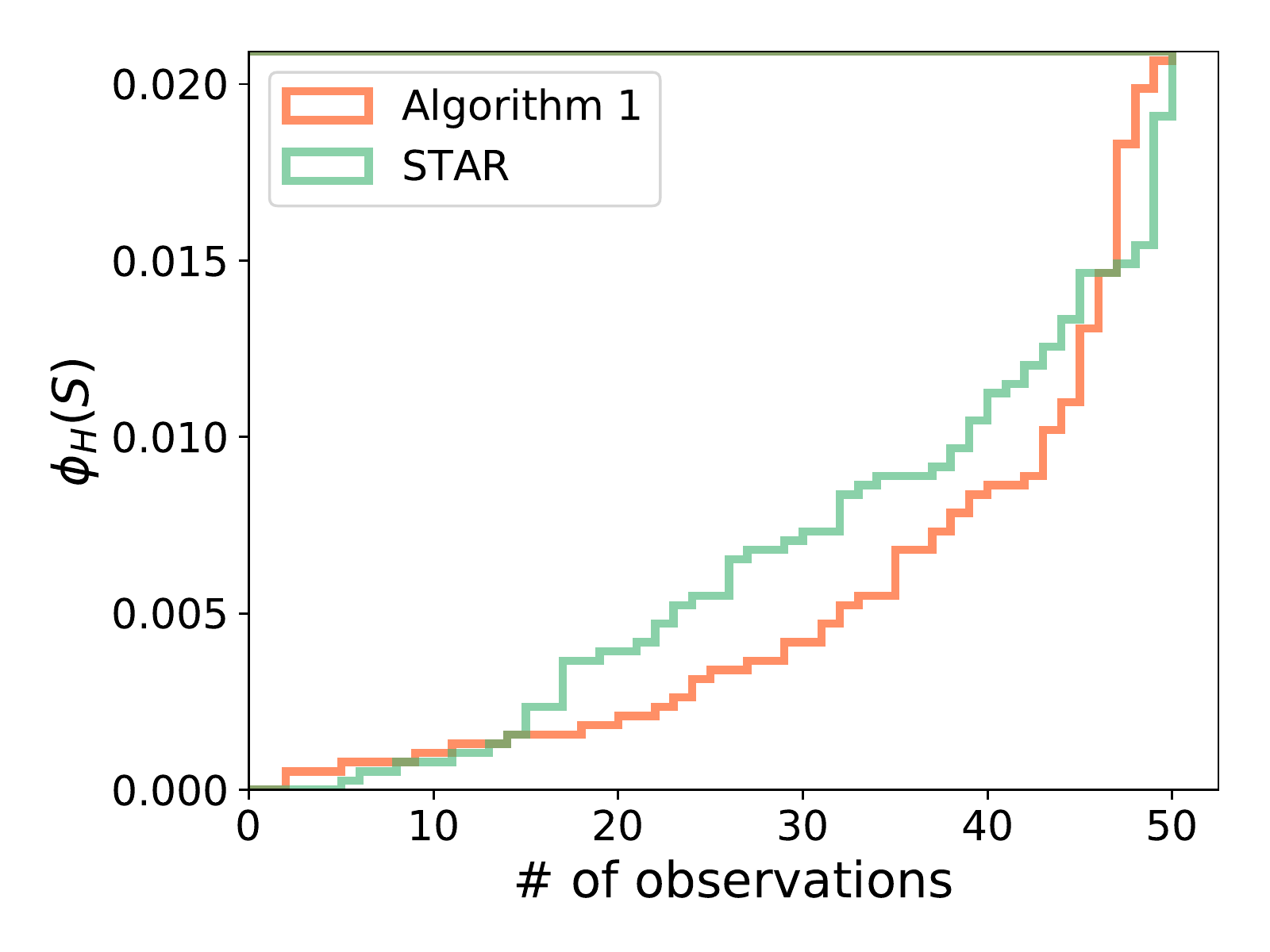}
}
\end{subfloat}
\caption{Comparison of the solution quality of Algorithm~\ref{alg:local-clustering} with those of \textsf{CLIQUE} and \textsf{STAR}.}\label{fig:local}
\end{figure*}

\paragraph{Baseline methods}
In local clustering, we used two baselines called \textsf{CLIQUE} and \textsf{STAR}.
\textsf{CLIQUE} first transforms a given hypergraph into a graph using the clique expansion,
i.e., for each hyperedge $e\in E$ and each $u, v\in e$ with $u\neq v$, it adds an edge $uv$ of weight $w(e)$.
Then \textsf{CLIQUE} computes personalized PageRank with some $\alpha$ using the power method and produces sweep cuts.
Finally, it outputs the best subset among them, in terms of the conductance in the original hypergraph.
\textsf{STAR} is an analogue of \textsf{CLIQUE}, which employs the star expansion alternatively,
i.e., for each hyperedge $e\in E$, it introduces a new vertex $v_e$, and then for each vertex $u\in e$, it adds an edge $uv_e$ of weight $w(e)/|e|$.
As the set of vertices has been changed in \textsf{STAR}, we leave out the additional vertices from the sweep cut.
As for the power method, the initial (PageRank) vector is set to be $(1/n, \dots, 1/n)$
and the stopping condition is set to be $\|x - x_\text{prev} \|_1\leq 10^{-8}$,
where $x$ and $x_\text{prev}$ are vectors obtained in the current and previous iterations, respectively.
Note that if we employ $A_\text{cand}$ as the candidate set of $\alpha$, the power method requires numerous number of iterations.
This is because $A_\text{cand}$ contains a lot of small values, 
e.g., $\frac{w_\mathrm{min}}{w_\mathrm{max}\sum_{e\in E}|e|}$.
Hence, we just used $\alpha=0.05,0.10,\dots, 0.95$ in our experiments.

In global clustering, we used the well-known hypergraph clustering package called \textsf{hMETIS}\footnote{
http://glaros.dtc.umn.edu/gkhome/metis/hmetis/overview
}
as well as the global counterparts of \textsf{CLIQUE} and \textsf{STAR}, which can be obtained in the same way as that of our algorithms.
Note that \textsf{hMETIS} has a lot of parameters; among those, the parameters \texttt{UBfactor}, \texttt{Nruns}, \texttt{CType}, \texttt{RType}, and \texttt{Vcycle} may affect the solution quality and running time.
To obtain a high quality solution, we run the algorithm for each of
$(\texttt{UBfactor}, \texttt{Nruns}, \texttt{CType}, \texttt{RType}, \texttt{Vcycle})
\in \{5, 10, \dots, 45\}\times \{10\}\times \{1,2,3,4,5\}\times \{1,2,3\}\times \{0,1,2,3\}$ to obtain vertex sets,
and then return the one with the smallest conductance.
We used \textsf{hMETIS} 1.5.3, the latest stable release.
Finally, we remark that \textsf{hMETIS} is not applicable to local clustering
because we cannot specify the volume of a subset containing the specified vertex.

\subsection{Local clustering}
First, we focus on local clustering and evaluate the performance of Algorithm~\ref{alg:local-clustering}.
To this end, for each hypergraph in Table~\ref{tab:instance}, we selected 50 initial vertices uniformly at random,
and for each of which, we ran Algorithm~\ref{alg:local-clustering}, \textsf{CLIQUE}, and \textsf{STAR}.
To focus on local structure, Algorithm~\ref{alg:local-clustering} employs $\mu=1/10$,
and so do \textsf{CLIQUE} and \textsf{STAR} in themselves.
The code is written in C\#, which was conducted on a machine with 2.8~GHz CPU and 16~GB RAM.

The quality of solutions are illustrated in Figure~\ref{fig:local}.
The results for \textsf{CLIQUE} are missing for the three largest hypergraphs because it could not be performed due to the memory limit.
In those plots, once we fix a value in the vertical axis,
we can see the cumulative number of solutions with the conductance less than or equal to the value fixed;
thus, an algorithm drawing a line with a lower right position has a better performance.
As can be seen, Algorithm~\ref{alg:local-clustering} outperforms \textsf{CLIQUE} and \textsf{STAR}.
In fact, Algorithm~\ref{alg:local-clustering} almost always has a larger cumulative number with a fixed conductance value.
It can be observed that \textsf{CLIQUE} performs better than \textsf{STAR}.

Table~\ref{tab:local_time} presents the running time of the algorithms, where we have the average values over 50 initial vertices.
The best results among the algorithms are written in bold.
As is evident, Algorithm~\ref{alg:local-clustering} is much more scalable than \textsf{CLIQUE} and comparable to \textsf{STAR}.

\begin{table}[t]
\begin{center}
\caption{Running time(s) of local algorithms. }\label{tab:local_time}
\scalebox{0.96}{
\begin{tabular}{lrrrr}
\toprule
Name & Algorithm~\ref{alg:local-clustering}   & \textsf{CLIQUE}   & \textsf{STAR}\\
\midrule
Graph products         &0.02    &\textbf{0.01}    &0.03                   \\
Network theory         &\textbf{0.06}    &0.18    &0.11                   \\
DBLP KDD               &\textbf{1.42}   &1.84       &1.54  \\
arXiv cond-mat              &5.43   &\textbf{4.84}   &5.70                   \\
DBpedia Writers             &\textbf{7.07}   &77.98   &17.50   \\
YouTube Group Memberships   &93.67   &OM   &\textbf{57.64}   \\
DBpedia Record Labels       &\textbf{36.69}   &OM   &44.18   \\
DBpedia Genre               &\textbf{85.86}   &OM   &97.47   \\
\bottomrule
\end{tabular}
}
\end{center}
\end{table}

\begin{table*}[t]
\begin{center}
\caption{Performance of Algorithm~\ref{alg:global-clustering}, \textsf{CLIQUE}, \textsf{STAR}, and \textsf{hMETIS}.}\label{tab:global}
\scalebox{1.0}{
\begin{tabular}{lrrrrrrrrrrrrrrr}
\toprule
&\multicolumn{2}{c}{Algorithm~\ref{alg:global-clustering}}
&&\multicolumn{2}{c}{\textsf{CLIQUE}}
&&\multicolumn{2}{c}{\textsf{STAR}}
&&\multicolumn{2}{c}{\textsf{hMETIS}}
\\
\cline{2-3}
\cline{5-6}
\cline{8-9}
\cline{11-12}
Name&$\phi_H(S)$ &Time(s)
&&$\phi_H(S)$ &Time(s)
&&$\phi_H(S)$ &Time(s)
&&$\phi_H(S)$ &Time(s)\\
\midrule
Graph products              &\textbf{0.0234}     &1.01        &&0.0246    &\textbf{0.70}       &&0.0288   &1.47         &&\textbf{0.0234}      &15.30 \\
Network theory              &0.0064     &\textbf{3.03}        &&0.0062    &8.57       &&0.0074   &5.15         &&\textbf{0.0061}      &24.82  \\
DBLP KDD                    &\textbf{0.0116}     &\textbf{68.91}       &&\textbf{0.0116}    &89.24      &&\textbf{0.0116}   &74.98        &&0.0236      &107.03    \\
arXiv cond-mat              &\textbf{0.0127}     &251.81      &&\textbf{0.0127}    &\textbf{237.31}     &&\textbf{0.0127}   &281.37       &&0.0259      &418.08      \\
DBpedia Writers             &0.0069     &\textbf{353.60}      &&\textbf{0.0043}    &3,812.72   &&0.0048   &856.85       &&0.0010      &721.36       \\
YouTube Group Memberships   &\textbf{0.0003}     &4,735.28    &&---       &OM         &&\textbf{0.0003}   &\textbf{2,779.52}     &&0.0029       &3,696.24\\
DBpedia Record Labels       &\textbf{0.0002}     &\textbf{1,766.34}    &&---       &OM         &&\textbf{0.0002}   &2,151.80     &&0.0052      &1,655.88 \\
DBpedia Genre               &\textbf{0.0001}     &\textbf{4,320.72}    &&---       &OM         &&\textbf{0.0001}   &4,779.95     &&0.0023      &8,670.49    &  \\
\bottomrule
\end{tabular}
}
\end{center}
\end{table*}

\subsection{Global clustering}
Next we focus on global clustering and evaluate the performance of Algorithm~\ref{alg:global-clustering}.
As can be inferred from Table~\ref{tab:local_time}, Algorithm~\ref{alg:global-clustering} itself does not scale for large hypergraphs.
Therefore, we modified the algorithm as follows: we selected 50 initial vertices uniformly at random,
and replace $V$ with the set of selected vertices in Line~\ref{line:global_V} (in Algorithm~\ref{alg:global-clustering}).
We modified the global counterparts of \textsf{CLIQUE} and \textsf{STAR} in the same way.
To compare the performance of the algorithms fairly, we used the same initial vertices for those three algorithms.
The code is again written in C\#, which was conducted on the above-mentioned machine.

The results are summarized in Table~\ref{tab:global}, where the best results are again written in bold.
As for the solution quality, Algorithm~\ref{alg:global-clustering} is one of the best choices,
but unlike the local setting, \textsf{STAR} performs well.
In fact, \textsf{STAR} achieves the same performance as that of Algorithm~\ref{alg:global-clustering} for the three largest hypergraphs.
On the other hand, the performance of \textsf{hMETIS} is unfavorable.
In our experiments, \textsf{hMETIS} was performed with almost all possible parameter settings;
therefore, it is unlikely that we can obtain a comparable solution using \textsf{hMETIS}.
The trend of the results for the running time is similar to that of the local setting;
Algorithm~\ref{alg:local-clustering} is much more scalable than \textsf{CLIQUE} and comparable to \textsf{STAR}.






\section{Conclusion}\label{sec:conclusion}

In this work, we have developed local and global clustering algorithms based on PageRank for hypergraphs, recently introduced by Li and Milenkovic~\cite{Li:2018we}.
To the best of our knowledge, ours are the first practical algorithms that have theoretical guarantees on the conductance of the output vertex set.
By experiment, we confirmed the usefulness of our clustering algorithms in practice.

\begin{acks}
The authors thank Yu Kitabeppu for helpful comments. 
This work was done while A.M. was at RIKEN AIP, Japan. 
This work was supported by JSPS KAKENHI Grant Numbers 18H05291 and 19K20218. 
\end{acks}

\bibliographystyle{ACM-Reference-Format}
\bibliography{main}

\appendix



\section{Proofs of Section~\ref{sec:basic-properties}}
\label{sec:proof-basic-properties}

   
\begin{proof}[Proof of Theorem~\ref{thm:existence-of-PR}]
  We first prove the former claim.
 In \cite{ikeda2018finding}, 
 it was proven that $\ca{L}_H$ is a maximal monotone operator. (We omit the definition here because we do not need it).
  Then, by the general theory of maximal monotone operators~\cite{miyadera1992nonlinear,komura1967nonlinear}, the operator $I + \lambda\ca{L}_H$ is injective for any $\lambda>0$.
  Hence, its inverse $J_\lambda := {(I+\lambda \ca{L}_H)}^{-1}$ is single-valued, and is called a \emph{resolvent} of $\ca{L}_H$.
  Because the image of $\ca{L}_H$ is $\RR^V$, for any $\bm{s}\in \RR^V$, there exists a unique vector $J_\lambda(\bm{s})$, which is a solution to~\eqref{eq:hypergraph-pagerank-via-laplacian} when
  $\lambda = (1-\alpha)/2\alpha > 0$.

  We now prove the latter claim.
  For $\bm{p} = \pr_\alpha(\bm{s})$, we have
  \[
    \bm{p} + \frac{1-\alpha}{2\alpha}\ca{L}_{H_{\bm{p}}}\bm{p} = \bm{s},
  \]
  where $H_{\bm{p}}$ is the graph obtained from $H$ and $\bm{p}$, as defined in Section~\ref{subsec:pagerank-hypergraph}.
  This implies that $\bm{p}$ is a classical PPR on $H_{\bm{p}}$, and hence the claim holds by the Perron--Frobenius theorem for
  the matrix $\alpha\bm{s}\1^\top + (1-\alpha)W_{H_{\bm{p}}}$.
\end{proof}


\begin{proof}[Proof of Propostion~\ref{prop:continuity}]
As mentioned in the proof of Theorem~\ref{thm:existence-of-PR}, the PPR $\bm{\pr}_\alpha(\bm{s})$ is equal to a resolvent $J_\lambda(\bm{s})$ for $\lambda = (1-\alpha)/2\alpha$.
Thus, it suffices to prove the continuity of the resolvent $J_\lambda(\bm{s})$ with respect to $\lambda>0$.
By~\cite{miyadera1992nonlinear},
for $\lambda \geq \mu>0$,
the resolvents satisfy the following inequality: 
\[
 J_\lambda(\bm{s}) = J_{\mu}\left( \frac{\mu}{\lambda}\bm{s}
 + \frac{\lambda-\mu}{\lambda} J_{\lambda}(\bm{s})  \right).
\]
On the other hand, it is well known that $J_\lambda$ is a non-expansive map, i.e., for any $\bm{s}, \bm{t} \in \RR^V$ and any
$\lambda > 0$,
$ \|J_\lambda(\bm{s}) - J_\lambda(\bm{t}) \|_{D_H^{-1}} \leq
 \|\bm{s} - \bm{t}\|_{D_H^{-1}}$, 
where $\| \bm{x} \|_{D_H^{-1}} =
\sqrt{\bm{x}^\top D_H^{-1} \bm{x}}$.
By these relations, we can show the following:
\begin{align*}
 & \|J_\lambda(\bm{s}) - J_\mu(\bm{s}) \|_{D_H^{-1}}
 \leq \left\|J_\mu\left(\frac{\mu}{\lambda}\bm{s} + \frac{\lambda-\mu}{\lambda}J_\lambda(\bm{s})\right) - J_\mu(\bm{s})  \right\|_{D_H^{-1}}
 \\
 &\leq \left\|\frac{\mu}{\lambda}\bm{s} + \frac{\lambda-\mu}{\lambda}J_\lambda(\bm{s}) - \bm{s}  \right\|_{D_H^{-1}}
 \leq \frac{\lambda-\mu}{\lambda} \|J_\lambda(\bm{s}) - \bm{s}
 \|_{D_H^{-1}}.
\end{align*}
Hence, $\lim_{\lambda \to \mu} J_\lambda(\bm{s}) = J_\mu(\bm{s})$.
\qedhere
\end{proof}

To prove Proposition~\ref{prop:bounds-of-values-of-PPR}, 
we prepare the following lemma: 
\begin{lem}\label{lem:pagerank-in-terms-of-beta}
  For $\beta:=2\alpha/(1+\alpha)$, we have
  \[
    \beta (\bm{s} - \bm{\pr}_\alpha(\bm{s})) - (1-\beta) \ca{L}_H (\bm{\pr}_\alpha(\bm{s})) = \bm{0}.
  \]
\end{lem}


\begin{proof}[Proof of Lemma~\ref{lem:pagerank-in-terms-of-beta} ]
  Let $\bm{p} = \pr_\alpha(\bm{s})$.
  Note that $(1-\alpha)/2\alpha = 1/\beta-1$.
  Then by~\eqref{eq:hypergraph-pagerank-via-laplacian} and easy calculation, we have
  \begin{align*}
    \left(I + \frac{1-\alpha}{2\alpha}\ca{L}_H\right)(\bm{p}) \ni \bm{s}
    \Leftrightarrow\; \beta (\bm{s} - \bm{p}) - (1-\beta) \ca{L}_H (\bm{p}) \ni \bm{0}.
  \end{align*}
  As $\bm{p}$ is unique by Theorem~\ref{thm:existence-of-PR}, we can replace the inclusion with equality.
\end{proof}

\begin{proof}[Proof of Proposition~\ref{prop:bounds-of-values-of-PPR}]
  Let $\bm{p} = \pr_\alpha(\bm{\chi}_u)$.
  Because $\bm{p}$ is a distribution by the latter claim of Theorem~\ref{thm:existence-of-PR}, we have
  $\bm{p}(v) \geq 0$ for any $v\in V$.
  Then by Lemma~\ref{lem:pagerank-in-terms-of-beta}, we have
  \[
   \bm{p} = \bm{\chi}_u - \frac{1-\beta}{\beta}
   \ca{L}_{H_{\bm{p}}} \bm{p} \geq \bm{0}.
  \]

  Hence for $v \neq u$, we have $-\ca{L}_{H_{\bm{p}}}\bm{p}(v) \geq 0$.
  Note that $-\ca{L}_{H_{\bm{p}}}\bm{p}(v)$ is equal to the difference between the heat diffusing into $v$ and that diffusing out from $v$ in the heat equation:
  \[
   \frac{d\bm{\rho}_t}{dt} = -\ca{L}_{H_{\bm{p}}}\bm{\rho}_t.
  \]
  Because the heat diffuses to the stationary distribution $\bm{\pi}_V$, the non-negativity of $-\ca{L}_{H_{\bm{p}}}\bm{p}(v)$ implies $\bm{p}(v) \leq \bm{\pi}_V(v)$.

  When $v = u$, we have $\bm{p}(u)-\bm{\chi}_u(u) \leq 0$.
  By a similar argument, we have the desired inequality.
\end{proof}


\section{Proof of Lemma~\ref{lem:upperbound-of-sweep-conductance-local-2}}\label{sec:proof-of-key-lemma}


\subsection{Setup}
Before proving
Lemma~\ref{lem:upperbound-of-sweep-conductance-local-2}, we first introduce several notations.
Let $G = (V, E, w)$ be an undirected graph.
For a subset $S\subseteq V$, we define two sets of
ordered pairs of vertices as follows:
\begin{align*}
 \mathrm{in}(S) := \{ (u,v) \in V\times S \}, \ 
 \mathrm{out}(S) := \{ (u,v) \in S \times V \}. 
\end{align*}
For an ordered pair $(u,v) \in V\times V$ and a distribution $\bm{p}\in \RR^V$,
we set $\bm{p}(u,v) := \frac{w(uv)}{d_u}\bm{p}(u)$.
For a set of ordered pairs of vertices $A \subseteq V\times V$ and
any distribution $\bm{p} \in \mathbb{R}^V$, we set $\bm{p}(A) := \sum_{(u,v)\in A} \bm{p}(u,v)$.
Then, by a similar argument for $\bm{p}$ and $H_{\bm{p}}$ 
to \cite{Andersen:2007tsa}, we have the following: 

\begin{lem}\label{lem:lazy-in-out}
  Let $H=(V,E,w)$ be a hypergraph, $\bm{p} \in \RR^V$ be a distribution, $H_{\bm{p}}=(V, E_{\bm{p}}, w_{\bm{p}})$ be a graph as defined in Section~\ref{sec:pre},
  and $S\subseteq V$.
  Then, we have
  \[
    W_{H_{\bm{p}}}(\bm{p})(S)= \frac{1}{2}\left( \bm{p}(\mathrm{in}(S) \cup \mathrm{out}(S)) + \bm{p}(\mathrm{in} (S)\cap \mathrm{out}(S)) \right),
  \]
  where $\bm{p}(\cdot)$ is defined using the graph $H_{\bm{p}}$.
\end{lem}


\subsection{Lov\'asz--Simonovic function for weighted graphs}
Next, we introduce a weighted version of the Lov\'asz--Simonovic function $\bm{p}[x]$~\cite{Lovasz90,Lovasz93}.
For a distribution $\bm{p} \in \RR^V$, we define $\bm{p}[x]$ on $x\in [0, \vol(V)]$ as follows:
If $x = \vol(S_j^{\bm{p}})$ for some $j = 0,1,\ldots,n$, we define
$\bm{p}[x] := \bm{p}(S_j^{\bm{p}})$.
For general $x \in [0, \vol(V)]$, we extend $\bm{p}[x]$ as a piecewise
 linear function. More specifically, if $x$ satisfies
 $\vol(S_j^{\bm{p}}) \leq x < \vol(S_{j+1}^{\bm{p}})$, then $\bm{p}[x]$ is defined as
\begin{align*}
 \bm{p}[x]
  = \bm{p}(S_j^{\bm{p}}) + \frac{\bm{p}(v_{j+1})}{d_{v_{j+1}}} (x - \vol(S_j^{\bm{p}})).
\end{align*}
We remark that the function $\bm{p}[x]$ is concave.
Then, an argument similar to that of~\cite{Andersen:2007tsa} with Lemma~\ref{lem:lazy-in-out} implies the following lemma:

\begin{lem}\label{lem:ineq-Lovasz-Simonovitz}
  For $\bm{p} = \bm{\pr}_\alpha(\bm{s})$ and $j = 1, 2, \dots, n-1$,
we have
  \begin{align*}
    & \bm{p}[\vol(S_j^{\bm{p}})]
    \leq \alpha \bm{s} \left[\vol(S_j^{\bm{p}})\right] + \\
    & \qquad + (1-\alpha)\frac{1}{2}\left(\bm{p}\left[\vol(S_j^{\bm{p}}) - |\partial(S_j^{\bm{p}})|\right] + \bm{p}\left[\vol(S_j^{\bm{p}}) + |\partial(S_j^{\bm{p}})|\right]\right).
  \end{align*}
\end{lem}

\subsection{A mixing result for PPR and the proof of Lemma~\ref{lem:upperbound-of-sweep-conductance-local-2}}
Next, we show a local version of a mixing result for PPR\@.
A global version for graphs was proven in~\cite{Andersen:2007tsa}.
\begin{thm}\label{thm:mixing}
  Let $\bm{s}$ be a distribution and $\phi \in [0,1]$,
  $\mu \in (0,1/2]$ be any constants.
  We set $\bm{p} = \pr_\alpha(\bm{s})$. Let $\ell_\mu \in \{1,2,\dots, n\}$ be the unique integer such that
$
    \vol(S_{\ell_\mu-1}^{\bm{p}}) < \mu\cdot\vol(V) \leq \vol(S_{\ell_\mu}^{\bm{p}}).$
 We assume that $\phi_H(S_j^{\bm{p}})\geq \phi$ holds for any $j = 1,2,\dots,\ell_\mu$. Then, for any $S \subseteq V$ such that $\vol(S)/\vol(V)\leq \mu$ and any $t\in \ZZ_+$, the following  holds:
\begin{align}
   \bm{p}(S) -\bm{\pi}_V(S) \leq \alpha t + \sqrt{\vol(S)} {\left( 1- \frac{\phi^2}{8} \right)}^t.  \label{eq:mixing}
\end{align}
\end{thm}

\begin{proof}
We can show this theorem by an argument similar to that of~\cite{Andersen:2007tsa}. Here, we show a sketch of the proof.

We define a function $f_t(x)$ on $x \in [0, \vol(S_{\ell_\mu}^{\bm{p}})]$ by
\[
 f_t(x) := \alpha t + \sqrt{ \min\{x, \vol(V) - x\}}
 {\left( 1- \frac{\phi^2}{8} \right)}^t.
\]
Because $\bm{p}(S) \leq \bm{p}[\vol(S)]$ for any vertex set $S \subseteq V$, if for a vertex set $S \subseteq V$ such that
$\vol(S)/ \vol(V) \leq \mu$,
\begin{align}
 \bm{p}[x] - \frac{x}{\vol(V)} \leq f_t(x) \label{eq:mixing-Lovasz}
\end{align}
holds for $x = \vol(S)$, then the inequality~\eqref{eq:mixing} for $S$
follows.

To prove Theorem~\ref{thm:mixing}, by induction on $t$,
we show that
if a sweep cut $S_j^{\bm{p}}$ satisfies $\phi_H(S_j^{\bm{p}})\geq \phi$,
then the inequality~\eqref{eq:mixing-Lovasz} holds for
 $x_j=\vol(S_j^{\bm{p}})$ and any $t\in\ZZ_+$ as follows:
\begin{enumerate}
  \leftskip=-5pt
 \item We prove that the inequality~\eqref{eq:mixing-Lovasz} holds for $t=0$ and any $S\subseteq V$.
 \item We assume that the inequality~\eqref{eq:mixing-Lovasz}
 holds for $t = t_0$ and $x_j = \vol(S_j^{\bm{p}})$.
 Then, by Lemma~\ref{lem:ineq-Lovasz-Simonovitz}, the concavity of $f_t(x_j)$, and the piecewise linearity of $\bm{p}[x_j]-x_j/\vol(V)$,
 we prove that if $S_j^{\bm{p}}$ satisfies
 $\phi_H(S_j^{\bm{p}})\geq \phi$,
then the inequality~\eqref{eq:mixing-Lovasz} holds for
 $x_j = \mathrm{vol}(S_j^{\bm{p}})$ and $t = t_0 + 1$.
 \end{enumerate}
This argument implies that if any sweep cut $S_j^{\bm{p}}$ ($j=1,2,\dots, \ell_\mu$) satisfies $\phi_H(S_j^{\bm{p}})\geq \phi$,
the inequality~\eqref{eq:mixing-Lovasz} holds for
$x_j=\vol(S_j^{\bm{p}})$ for all $j \in \{0,1,\dots, \ell_\mu \}$ and any $t\in \ZZ_+$.
Then, by the concavity of $f_t(x)$ and the piecewise
 linearity of $p[x] - x/\vol(V)$ again, the inequality~\eqref{eq:mixing-Lovasz} holds for any $x \in [0, \mu\cdot \vol(V)]$ and any $t\in \ZZ_+$.
 Because $\bm{p}(S) \leq \bm{p}[\vol(S)]$, we obtain the  inequality~\eqref{eq:mixing} for any
 $S\subseteq V$ with $\vol(S)/\vol(V)\leq \mu$. \qedhere
\end{proof}

We next prove the following  
as a consequence of Theorem~\ref{thm:mixing}.

\begin{lem}\label{lem:upperbound-of-sweep-conductance-local}
  If there is a vertex set $S \subseteq V$ and
  a constant
  $\delta \geq 4/\sqrt{\vol(V)}$ such that
  $\vol(S)/\vol(V) \leq \mu$ and $\bm{\pr}_\alpha(\bm{s}) (S) -\bm{\pi}_V(S) > \delta$,
  then the following inequality holds:
  \begin{align}
    \phi_H^\mu(\bm{\pr}_{\alpha}(\bm{s})) <
    \sqrt{\frac{12 \alpha \log(\vol(V))}{\delta}}.\label{eq:sweep-bounded-by-alpha}
  \end{align}
\end{lem}

\begin{proof}
We apply Theorem~\ref{thm:mixing} for $\phi = \phi_H^\mu(\pr_\alpha(\bm{s}))$. Then,
because every sweep cut $S_j^{\bm{p}}$
($j = 1,2, \dots, \ell$) satisfies the assumption 
$\phi_H(S_j^{\bm{p}}) \geq \phi$ in Theorem~\ref{thm:mixing},
 any vertex set $S\subseteq V$ with $\vol(S)/\vol(V)\leq \mu$ satisfies
 the inequality~\eqref{eq:mixing} for any $t \in \ZZ_+$.
As in the argument of~\cite{Andersen:2007tsa}, we take $t$ as
\[
 t = \left\lceil \frac{8}{\phi^2}\log\left(\frac{4\sqrt{\vol(S)}}{\delta}
 \right) \right\rceil \leq \frac{9}{\phi^2}\log\left(\frac{4\sqrt{\vol(S)}}{\delta}
 \right).
\]
Here, $\lceil x \rceil = \min\{ a \in \ZZ \mid x \leq a \}$.
Then by the inequality~\eqref{eq:mixing}, we have
\[
 \bm{p}(S) -\bm{\pi}_V(S) \leq \alpha \frac{9}{\phi^2}\log\left(\frac{4\sqrt{\vol(S)}}{\delta}
 \right)  + \frac{\delta}{4}.
\]
By the assumption $\delta < \bm{p}(S) -\bm{\pi}_V(S)$
and rearranging the obtained inequality, we have the desired inequality~\eqref{eq:sweep-bounded-by-alpha}. \qedhere
\end{proof}

Finally, we deduce Lemma~\ref{lem:upperbound-of-sweep-conductance-local-2} from Lemma~\ref{lem:upperbound-of-sweep-conductance-local}. 
\begin{proof}
[Proof of Lemma~\ref{lem:upperbound-of-sweep-conductance-local-2}]
The statement of Lemma~\ref{lem:upperbound-of-sweep-conductance-local} is independent of the normalization of the weight function $w\colon E \to \RR_{\geq 0}$, except for the factor $\vol(V)$. Hence, for the hypergraph $H^c=(V,E,cw)$ scaled by  $c >0$, the similar statement to Lemma~\ref{lem:upperbound-of-sweep-conductance-local} obtained by replacing $\vol(V)$ with 
$c \cdot \vol(V)$ also holds. 
This means that the right hand side of 
\eqref{eq:sweep-bounded-by-alpha} can be decreased as long as 
the assumption $\delta \geq 4/\sqrt{c\cdot\vol(V)}$ holds. 
Then, we remark that the obtained conductance 
$\phi_{H^c}^\mu (\pr_\alpha(\bm{s}))$ is independent of $c$. 
 As a consequence of this argument, if a hypergraph $H$ satisfies 
 the assumption of Lemma~\ref{lem:upperbound-of-sweep-conductance-local}, then, Lemma~\ref{lem:upperbound-of-sweep-conductance-local} holds for $H^c$ for any $c$ that satisfies the assumption $\delta \geq 4/\sqrt{c \cdot \vol(V)}$, hence 
also for $c\cdot\vol(V) = 4^2/\delta^2$. 
However, the left hand side of 
\eqref{eq:sweep-bounded-by-alpha} for $H$ is the same as that 
for $H^c$. 
This implies the inequality \eqref{eq:sweep-bounded-by-alpha-2}.
\end{proof}


\section{Proof of Lemma~\ref{lem:lowerbound-of-pagerank}}
\label{sec:proof-lowerbound-of-pagerank}
\begin{proof}[Proof of Lemma~\ref{lem:lowerbound-of-pagerank}]
Let $\bm{p}_v = \bm{\pr}_\alpha(\bm{\chi}_v)$ for $v\in V$.
We consider a random variable $X := \bm{p}_v(\ol{C})$, where $v \in V$ is sampled according to the distribution $\bm{\pi_C}$.
Then by Lemma~\ref{lem:bound-by-conductance}, we have
\begin{align*}
 \mathrm{E}_{\bm{\pi}_C}[X] & = \sum_{v \in V} \bm{\pi}_C(v) \bm{p}_v(\ol{C})
 \leq \frac{\phi_H(C)}{2\alpha}.
\end{align*}
By Markov's inequality, we have
\[
 \Pr[v \not\in C_\alpha] \leq  \Pr\left[ X > 2
 \mathrm{E}_{\bm{\pi}_C}[X]\right] \leq \frac{1}{2}.
\]
This implies that $\vol(C_\alpha)$ is larger than $\vol(C)/2$.\qedhere
\end{proof}

\section{Proofs of sufficient conditions}\label{sec:proof-sufficient-conditions}

\begin{proof}[Proof of Lemma~\ref{lem:criterion-1}]
Let $\bm{p}_v = \bm{\pr}_\alpha(\bm{\chi}_v)$ for $v \in V$.
For a vertex $v \in V$, we have
\[
  \bm{p}_v(v) \leq 
  \frac{1}{2} 
  \leq 1- \frac{\vol(C)}{\vol(V)}.
\]
Hence by multiplying both sides by $\bm{\pi}_C(v) = d_v/\vol(C)$ and 
transposing a term, we have 
\begin{align}
 \bm{\pi}_C(v)\bm{p}_v(v) +
  \bm{\pi}_V(v)  \leq \bm{\pi}_C(v). 
  \label{eq:inequality-pr-unif}
\end{align}
Now, we have
\begin{align*}
 & \sum_{u\in C}\bm{\pi}_C(u)\bm{p}_u(v)
 = \bm{\pi}_C(v)\bm{p}_v(v) +
 \sum_{u\in C\setminus \{v \}}\bm{\pi}_C(u)
 \bm{p}_u(v) \\
 &\leq \bm{\pi}_C(v)\bm{p}_v(v) +
 \sum_{u\in C\setminus \{v \}}\bm{\pi}_C(u)
 \bm{\pi}_V(v) \tag{by Proposition~\ref{prop:bounds-of-values-of-PPR}}\\
 &\leq \bm{\pi}_C(v)\bm{p}_v(v) +
 \bm{\pi}_V(v) \leq \bm{\pi}_C(v) \tag{by~\eqref{eq:inequality-pr-unif}}.
\end{align*}
This concludes the proof.
\end{proof}


\begin{proof}[Proof of Lemma~\ref{lem:criterion-2}]
  Let $\bm{p} = \bm{\pr}_\alpha(\bm{\chi}_v)$.
  By an argument similar to the proof of Lemma~\ref{lem:bound-by-conductance} for the case $C = \{v \}$, we have
  \begin{align*}
  & \alpha(1 - \bm{p}(v)) = (1-\alpha) \bm{p}(v) - (1-\alpha)
  \sum_{u \in \ol{C}, u\neq v} \frac{w_{\bm{p}}(vu)}{d_{u}} \bm{p}(u).
  \end{align*}
  Then, we have
  \[
    \bm{p}(v)
    = \alpha + (1-\alpha) \sum_{u \in \ol{C}, u\neq v} \frac{w_{\bm{p}}(vu)}{d_{u}} \bm{p}(u).
  \]
  Because $\bm{p}(u) \leq d_u/\vol(V)$ holds for any $u\neq v$ by Proposition~\ref{prop:bounds-of-values-of-PPR},
  we have
  \begin{align*}
  \bm{p}(v) \leq \alpha + (1-\alpha)
  \sum_{u \in \ol{C}, u\neq v} \frac{w_{\bm{p}}(vu)}{d_{u}} \frac{d_u}{\vol(V)}
  \leq \alpha + (1-\alpha)\frac{d_{\max}}{\vol(V)}
  \end{align*}
  By combining this inequality with Lemma~\ref{lem:criterion-1}, if $\alpha + (1-\alpha)d_{\max}/\vol(V) \leq 1/2$, then $\alpha$ satisfies
  the condition~\eqref{eq:assumption} for any
$v\in C\!\setminus\! C^\circ$.
\end{proof}


\end{document}